\documentclass[12pt,reqno]{amsart}
\usepackage[left=1.25in,right=1.25in,top=1.25in,bottom=1.25in]{geometry}
\usepackage{graphicx}
\usepackage{color}
\usepackage[dvipsnames]{xcolor}
\usepackage{amsmath, amsthm, amssymb, amsfonts}
\usepackage{accents}
\usepackage{natbib}
\usepackage{har2nat}
\usepackage{tabularx}
\usepackage{caption}
\usepackage{parskip}
\usepackage[british]{babel}
\usepackage{setspace}

\newtheorem{theorem}{Theorem}

\newtheorem{lemma}{Lemma}
\newtheorem{proposition}{Proposition}
\theoremstyle{definition}

\newtheorem{remark}{Remark}
\newcommand{\eps}{\varepsilon}
\newcommand{\ul}{\underaccent{\bar}}
\newcommand{\ol}{\bar}
\newcommand{\df}{\mathrm{d}}

\newcommand{\R}{\mathbb{R}}
\newcommand{\N}{\mathbb{N}}

\begin{document}

\let\MakeUppercase\relax

\title{Robust Sequential Search}
\author[\uppercase{Schlag and Zapechelnyuk}]{
\larger \textsc{Karl H.~Schlag and Andriy
Zapechelnyuk}}
\date{\today}
\thanks{ \ \\
\textit{Schlag}: Department of Economics, University of Vienna, Oskar-Morgenstern-Platz 1,
1090 Vienna, Austria. \emph{E-mail:} karl.schlag@univie.ac.at. \\
\textit{Zapechelnyuk}: School of Economics and Finance, University of St Andrews, Castlecliffe, the Scores, St Andrews KY16 9AR, UK. {\it E-mail:} {az48@st-andrews.ac.uk.} \\ 
\ \\
The authors would like to thank Dirk Bergemann, Jeffrey Ely, Olivier Gossner, Johannes H{\"o}rner, Bernhard Kasberger, and Alexei Parakhonyak for helpful comments and suggestions.
}

\begin{abstract} 
\onehalfspacing
We study sequential search without priors. Our interest lies in decision rules that are close to being optimal under each prior and after each history. We call these rules dynamically robust. The search literature employs optimal rules based on cutoff strategies that are not dynamically robust. We derive dynamically robust rules and show that their performance exceeds 1/2 of the optimum against binary environments and 1/4 of the optimum against all environments. This performance improves substantially with the outside option value, for instance, it exceeds 2/3 of the optimum if the outside option exceeds 1/6 of the highest possible alternative.

\bigskip

\noindent\emph{JEL\ Classification:}\ D83, D81, C44\newline


\

\end{abstract}

\maketitle
\onehalfspacing

\newpage
\section{Introduction}\label{s:intro}
Suppose that you check stores one by one in search of the cheapest 
place to buy some good. Your decision of when to stop searching depends 
on the distribution of prices you expect to encounter in unvisited 
stores. The methodology of Bayesian decision making proposes to turn 
this into an optimization problem, using as input your prior belief about 
possible distributions, mathematically formulated as a distribution over 
distributions. This is a complex and usually intractable intertemporal decision problem.
Special cases can be solvable, but solutions are fragile as they depend on your beliefs about what you do not know (see \citeauthor{Gastwirth1976}, 1976).

We are interested in a robust approach to this problem that does not depend on specific prior beliefs of a decision maker. Instead of focusing on optimality for some prior, we look for epsilon-optimality for all priors. Furthermore, we are interested in a dynamically consistent approach in which the performance matters at any point in time and not only at the outset. In this paper, we formalize a performance criterion that fulfills these desiderata. Decision rules that are optimal under this criterion are called {\it dynamically robust}. We derive general properties of dynamically robust rules and then show how close their performance is to the optimal ones under each prior.

The practical relevance of robust decision making is apparent.
How can a shopper know the distribution of prices offered in the next store? How does
she form a prior about such distributions? Even if a prior is formed, will
the shopper be able to overcome the complexity of Bayesian optimization? Will the decision rule still be good if the prior puts little
or no weight on the environment that is realized? How will the shopper
argue about the optimality of a particular decision rule in front of her
peers if they do not have the same prior as she does? These questions can be addressed by a decision rule that performs relatively well for any
prior. Such a rule can be proposed as a compromise among Bayesian decision makers who have
different priors. It is a shortcut to avoid cumbersome calculations involved in finding the
Bayesian optimal rule. Finally, as a single rule that does not depend on individual (unobservable) beliefs, it is a useful benchmark for empirical studies.

The setting we consider in this paper is as follows. Alternatives arrive according to some i.i.d.~process. An individual who does not know the underlying distribution has to decide after each draw whether to stop or to continue. There is free recall: when the
individual stops she can choose the best alternative found so far. Values are discounted over
time, thus, waiting for better alternatives is costly. In an extension we also include an additive cost of waiting for better alternatives (see Appendix \ref{s:cost}).

As our first contribution, we develop a methodology for robust decision making that applies not only to sequential search. In a nutshell, we replace optimality for a given prior by epsilon-optimality under all priors and after all histories of past observations, and look for a smallest such epsilon. 
Specifically, we measure the performance of a given decision rule as follows. For each prior and each history, we compute the ratio of the rule's payoff to the maximal possible payoff. We then evaluate the rule by the smallest of these ratios and call it the {\it performance ratio} of the rule. This performance ratio describes what fraction of the maximal payoff can be guaranteed, regardless of the prior under which the payoffs are computed and regardless of which alternatives have realized over time. We are interested in a decision rule that achieves the largest possible performance ratio. Such a rule will be called {\it dynamically robust}. 

As our second contribution, we solve the described sequential search
problem. This is done first for binary environments, and then for more general environments. An environment is {\it binary } if it is a lottery over two alternatives, low and high. The values of these alternatives need not be known to the individual. We find that the dynamically robust performance ratio against binary environments is at least 1/2. So, the individual can always guarantee at least a half of the maximal payoff, even if the value of the maximal payoff is not known. Moreover, if there is an upper bound on the possible values of the high alternative, then the dynamically robust performance ratio is strictly increasing in the individual's outside option, attaining 2/3 and 3/4 when the outside option is, respectively, 1/6 and 1/3 of that upper bound. Surprisingly, these results extend to general environments, provided that possible values of alternatives have an upper bound, and the outside option is not too small. The decision rule that supports these findings prescribes to stop after any given history with a probability that is increasing in the value of the best realized alternative. In general, we show that the dynamically robust performance ratio is always at least 1/4, where this lower bound is attained when alternatives are unbounded, or in the limit as the outside option approaches zero.

Our analysis reveals that a dynamically robust rule has three notable properties. 

First, any such rule prescribes randomization between stopping and continuing the search. Intuitively, one should not stop with certainty when concerned that future outcomes may be higher. Similarly, one should not continue with certainty when concerned that future outcomes may never be higher. This stands in contrast with almost the entire search literature that studies deterministic cutoff rules.\footnote{An exception is Janssen et al.~(\citeyear*{Janssen2017}).}  

Second, a dynamically robust rule does not make any inference about the environment from past observations. The reason is that, after any history of explored alternatives, some degenerate environments can be ruled out. Yet, for every such environment, there are arbitrarily close nondegenerate environments that cannot be ruled out. Thus the closure of set of feasible posteriors about unexplored alternatives remains unchanged.

Finally, the worst-case priors that determine the robust performance ratio are degenerate, assigning probability one to a specific i.i.d.~distribution. This means that the payoff ratio of a decision rule can only be higher under nondegenerate priors. In addition, the worst-case distributions have support on at most two different values of alternatives. Loosely speaking, this is because the individual makes a binary choice in each round, and hence, two values, high and low, provide enough freedom to construct worst-case distributions. 

Our dynamically robust rules can be replaced by simpler rules without substantially changing the performance ratio. These simpler rules involve a stopping probability that is linear in the best realized alternative (see Appendix \ref{s:linear}).

\subsection*{Alternative Approaches to Performance Measurement}
Our paper deals with decision making under multiple priors. A prominent candidate criterion is maximin expected utility, as in \citeasnoun{Wald50} and \citeasnoun{Gilboa89}. There is a conceptual reason why we do not follow this approach. In this paper, we maintain the classic utility maximization preferences, moving only from optimality to epsilon optimality. Our approach makes sense to one who is unable to solve the sequential decision problem, unsure which specific prior to assign, or in need of justifying behavior in front of others. In contrast, the maximin utility approach does not have any one of these interpretations. 
Moreover, it takes a very different approach to multiplicity of priors. Instead of trying to be good irrespective of the prior (as in the original meaning of the term ``robust'' as discussed below) it aims to do best for the very specific prior where payoffs are lowest. 

On top of this, the maximin utility approach is too restrictive in the sequential search problem. 
The rule selected by the maximin utility criterion prescribes to stop immediately and not to search at all.
So, this criterion does not present useful insights for understanding how to search.

Another criterion that receives a lot of attention is minimax regret. The degree of suboptimality (referred to as regret) is measured either in terms of differences \citep{Savage51} or, as popular in the computer science literature, in terms of ratios (\citeauthor{Sleator85}, 1985; see also the axiomatization of \citeauthor{Terlizzese}, 2008), which can also be found in the robust contract literature (e.g., \citeauthor{Chassang2013}, 2013). We prefer ratios to obtain a scale-free measure and, thus, to be able to compare the performance after different histories, as well as across different specifications of the environment.

A common feature in the minimax regret literature is the evaluation of the payoffs retrospectively, after all uncertainty is resolved, as in the search models of \citeasnoun{Bergemann2011WP} and \citeasnoun{Parakhonyak2015}. Instead, we adopt a forward-looking approach, similar to \citeasnoun{Hansen2001}, \citeasnoun{Perakis}, Jiang et al.~(\citeyear*{Jiang}), and \citeasnoun{Kasberger2017}. The individual judges and compares decision rules by their discounted expected payoffs before the uncertainty is resolved, as a standard Bayesian decision maker would. 

An innovative aspect to our methodology is that, in the spirit of Bayesian decision making, we evaluate the performance not only ex-ante, but also after each additional piece of information has been gathered. 
We identify a bound on the relative performance loss that the decision maker tolerates in exchange for having a rule that does not depend on a specific prior. The corresponding decision rule is dynamically consistent in the sense that this bound will not be exceeded, regardless of what alternatives are realized. We are not aware of any paper that either formulates or derives dynamically consistent robust search behavior.\footnote{\citeasnoun{Schlag2015} consider dynamic decision making without priors in a non-search setting. A crucial difference from this paper is that they compare the performance of a decision rule to those of a few given benchmark strategies, not to the optimal behavior for the underlying model.} In particular, ex-ante commitment is required in the literature on the secretary problem \citep{Fox1960} that studies sequential search within a nonrandom set of exchangeable alternatives (for a review, see \citeauthor{Ferguson1989}, 1989).\footnote{We investigate the secretary problem under our criterion of dynamic robustness in a separate paper \citep{SZ-Secretary}.}   An analysis of ex-ante robust search in the setting of this paper is difficult and remains unsolved. \citeasnoun{Bergemann2011WP} and \citeasnoun{Parakhonyak2015} study a special case with two periods, and Babaioff et al.~\citeyear{Babaioff2009} study asymptotic performance of approximately optimal algorithms in a related problem with no recall, so these results are not comparable to our paper.

\subsection*{Other Related Literature.}

The term {\it robustness} goes back to Huber (\citeyear*{Huber1964}, \citeyear*{Huber1965}), defined as a procedure whose
\textquotedblleft performance is insensitive to small deviations of the
actual situation from the idealized theoretical model\textquotedblright \ 
\citep{Huber1965}. \citeasnoun{Prasad2003} and \citeasnoun{Bergemann2011JET}
formalize this notion for a policy choice, they measure insensitivity under
small deviations as performance being close to that of the optimal policy. The same approach has been applied to large deviations, where the performance is evaluated under a large class of distributions, as in statistical treatment choice (\citeauthor{Manski2004}, 2004, \citeauthor{Schlag2006}, 2006, and \citeauthor{Stoye2009}, 2009), auctions \citep{Kasberger2017}, and search in markets (\citeauthor{Bergemann2011WP}, 2011b, and \citeauthor{Parakhonyak2015}, 2015). The term {\it robustness} has been used in the same spirit -- to achieve
an objective independently of modeling details -- in robust
mechanism design \citep{Bergemann05}, and in the
field of control theory (Zhou et al., \citeyear*{Zhou1995}).\footnote{The term {\it robustness} has also been used in other contexts. It appears in the maximin utility approach (\citeauthor{Wald50}, 1950, and \citeauthor{Gilboa89}, 1989) adapted to robust contract design (\citeauthor{Chassang2013}, 2013, and \citeauthor{Carroll2015}, 2015), robust optimization (Ben-Tal et al., \citeyear*{BenTal2009}), robust selling mechanisms (Carrasco et al., \citeyear*{Carrasco}), and robust control in macroeconomics \citep{Hansen2001}. It also appears in \citeasnoun{Kajii97} where the concept of robustness is related to closeness in the strategy space, rather than in the payoff space.}

Dynamic consistency has been studied in other models of choice under ambiguity by \citeasnoun{Epstein2003}, Maccheroni et al.~(\citeyear*{Maccheroni2006}), Klibanoff et al.~(\citeyear*{Klibanoff2009}), \citeasnoun{Riedel2009}, and \citeasnoun{Siniscalchi2011}.  The challenge in this literature has been how to appropriately update information over time. In many cases this can only be done by artificially constraining possible environments and priors. We avoid the resulting conceptual and technical obstacles by letting a Bayesian decision maker process the information, which is dynamically consistent by definition.

\section{Model}\label{s:model}

\subsection{Setting}
An individual chooses among alternatives that arrive sequentially. She starts with an outside option $x_0$ which is given and is strictly positive, so $x_0>0$. Alternatives $x_1,x_2,...$ are realizations of an infinite sequence of i.i.d.~random variables. Each $x_t\ge 0$ describes how much this alternative is worth to the individual. In each round $t=0,1,2,...$, after having observed $x_t$, the individual decides whether to stop the search, or to wait for another alternative. There is {\it free recall}: when the individual decides to stop, she chooses from all the alternatives she has seen so far. The highest alternative in a history $h_t=(x_0,x_1,...,x_t)$ is referred to as {\it best-so-far alternative} and denoted by $y_t$, so
\[
y_t=\max\{x_0,x_1,...,x_t\}.
\]
Payoffs are discounted over time with a discount factor $\delta\in (0,1)$. From the perspective of round 0, the payoff of stopping after $t$ rounds is $\delta^t y_t$. The discount factor incorporates various multiplicative costs of search, such as the individual's impatience and a decay of values that are not accepted.\footnote{The restriction to multiplicative search costs is for simplicity and clarity of exposition. Our methodology extends to more general costs of search that include both additive and multiplicative components, as we show in Appendix \ref{s:cost}.}  

We assume that alternatives are drawn from a given (Borel) set $X\subset\R_+$, with $0\in X$, 
according to a probability distribution $F$.\footnote{Inclusion of $0$ in $X$ is for notational convenience. Nothing changes if we replace 0 by some $\ul x$ as long as the outside option satisfies $x_0\ge\ul x$. Inclusion of 0 is natural in applications where search may not provide a new alternative in each round, so the absence of a new alternative is modeled as the zero-valued alternative.} For instance, the set of alternatives $X$ can be $\R_+$, $\N$, $[0,\bar x]$, or $\{0,\bar x\}$. Let $\mathcal F_X$ denote the set of all distributions over $X$ that have a finite mean.\footnote{Distribution $F$ must have a finite mean to ensure that the optimal payoff under $F$ is well defined.}  We refer to $F$ as an {\it environment} and to $\mathcal F\subset \mathcal F_X$ as a {\it set of feasible environments}.

We also allow for mixed environments. A {\it mixed environment} is a probability distribution with a finite support over the set of feasible environments $\mathcal F$.\footnote{We restrict attention to mixed environments with finite support to avoid technical complications of dealing with priors over infinite sets. In fact, we show later that the analysis reduces to dealing with pure environments only, so the restriction to finite support plays no role in the results.} The set of mixed environments is denoted by $\Delta(\mathcal F)$. Mixed environments capture applications where each alternative $x_t$ depends on two components, an independent value $\xi_t$ and a common value $\theta$. For example, in a job search model, the value $x_t$ of a job offer may be expressed as $x_t=\theta+\xi_t$, where $\theta$ is  a market-wide or jobseeker-specific unobservable variable, and $\xi_t$ is an idiosyncratic unobservable value specific to employer $t$. 

The decision making of the individual is formalized as follows. Clearly, if the individual stops, 
she chooses the best-so-far alternative. So, the only relevant decision is when to stop. This is given by a {\it decision rule} $p$ that prescribes for each history of alternatives $h_t=(x_0,x_1,...,x_t)$ the probability $p(h_t)$ of stopping after that history.

\subsection{Bayesian Decision Making}\label{s:Bayesian}

A Bayesian approach to this search problem is as follows. A Bayesian decision maker starts with some prior over (mixed) environments. In each round, she updates this prior according to Bayes' rule and makes a choice that maximizes her expected payoff under the current posterior. Given a prior $\mu$, we call such a decision rule {\it optimal under $\mu$}.  

Note that each prior, formally defined as a probability distribution with a finite support over the set mixed environments $\Delta(\mathcal F)$, is a compound lottery over the set of environments $\mathcal F$. Beca use compound lotteries are equivalent to simple lotteries, any prior over mixed environments, $\mu\in\Delta(\Delta(\mathcal F))$, is an element of $\Delta(\mathcal F)$ itself. In what follows, we will refer to elements of $\Delta(\mathcal F)$ synonymously as priors and mixed environments. 

A prior is called {\it degenerate} if it assigns unit mass on a single environment $F\in\mathcal F$. By convention, we associate each environment $F$ with the correspondent degenerate prior, so $F\in \Delta(\mathcal F)$.

An environment $F\in\mathcal F$ and a prior $\mu\in\Delta(\mathcal F)$ are called {\it consistent} with a history of alternatives $h_t=(x_0,x_1,...,x_t)$ if the sequence of alternatives $x_1,...,x_t$ occurs with a positive probability under $F$ and $\mu$, respectively. Denote by $\mathcal F(h_t)$ the sets of environments that are consistent with $h_t$. With a slight abuse of notation, denote by $\Delta(\mathcal F(h_t))$ the sets of priors that are consistent with $h_t$.

%

\subsection{Performance Criterion\label{s:criterion}} 
We consider an individual who does not know which environment she faces. Rather than being concerned with the optimality under a particular prior, the individual wishes to find a decision rule that is approximately optimal under all priors and at all stages of the decision making. We formalize this performance criterion as follows.

Consider a set of alternatives $X$, a set of feasible environments $\mathcal F\subset\mathcal F_X$, a history of alternatives $h_t=(x_0,x_1,...,x_t)$, and a prior $\mu\in \Delta(\mathcal F(h_t))$, so $\mu$ is consistent with history $h_t$. Let $U_p(\mu,h_t)$ denote the expected payoff of a decision rule $p$ under $\mu$, conditional on history $h_t$, so
\begin{align}
U_p(\mu,h_t) &=p(h_t)y_t+(1-p(h_t))\delta \int_{\mathcal F}\int_X U_p(\mu,h_t\oplus x_{t+1}) \df F(x_{t+1})\df \mu(F),\label{U-0}
\end{align}
where $h_t\oplus x_{t+1}=(x_0,...,x_t,x_{t+1})$. 

Let $V(\mu,h_t)$ denote the {\it optimal payoff} under $\mu$ conditional on $h_t$,
\begin{align*}
V(\mu,h_t)&=\sup\nolimits_{p} U_p(\mu,h_t).
\end{align*}
This is the highest possible expected payoff, in other words, the payoff of a Bayesian decision maker, under prior $\mu$ given history $h_t$.

The {\it payoff ratio} $U_p(\mu,h_t)/V(\mu,h_t)$ describes the fraction of the optimal payoff that a given rule $p$ attains under prior $\mu$ given history $h_t$. Note that $V(\mu,h_t)\ge x_0>0$.

The {\it performance ratio} $R_p(x_0,\mathcal F)$ of a decision rule $p$ is defined as the lowest payoff ratio over all histories of alternatives and all priors consistent with those histories,
\[
R_p(x_0,\mathcal F)=\inf_{h\in\mathcal H(x_0)}\inf_{\mu\in \Delta(\mathcal F(h))} \frac{U_p(\mu,h)}{V(\mu,h)},
\]
where $\mathcal H(x_0)$ denotes the set of histories with outside option $x_0$.
So, the performance ratio captures the fraction of the optimal payoff that a rule guarantees in each round. 

The highest possible performance ratio is called {\it dynamically robust} and is given by
\[
R^*(x_0,\mathcal F)=\sup\nolimits_{p} R_{p}(x_0,\mathcal F).
\]
Note that $R^*(x_0,\mathcal F)$ depends only on the information available from the start:  the outside option $x_0$, the set of feasible environments $\mathcal F$, and, implicitly, the discount factor $\delta$. 

A decision rule $p^*$ is called {\it dynamically robust} if it attains the dynamically robust performance ratio, so $R_{p^*}(x_0,\mathcal F)=R^*(x_0,\mathcal F)$.

\subsection{Motivation}
Our performance criterion can be motivated by the concept of epsilon-optimality. 
In this paper we replace the objective of optimality against a given environment or prior by the objective of epsilon-optimality against all environments and priors. Such a rule is {\it robust} in the sense that its performance remains close to the optimum irrespective of which particular environment in $\Delta(\mathcal F)$ the individual faces.

An important aspect of economic models of search is their dynamic nature. Decisions are made in each round, and past search costs are sunk, hence irrelevant for today's choices. This dynamic nature is an integral part of our approach. We are interested in {\it dynamic consistency} of a rule, in the sense that its epsilon-optimality should hold not only ex-ante, but also in all subsequent rounds. 
This is why we use the term {\it dynamically robust}.

The dynamic robustness criterion does not require a decision maker to be too specific
about the environment. It is appropriate for a decision maker who is willing to sacrifice payoffs in favor of more general applicability and performance stability. Imagine an individual (e.g., a CEO of a company) who must convince a group of observers (e.g., a board of directors), each with a different prior, that her decision rule is good. Assume that these observers can monitor the performance of this decision rule over time, so they must remain convinced at all stages of the decision making. If the individual's rule is dynamically robust, then no observer will ever be able to accuse the individual of underperforming by more than a specified threshold. Moreover, being dynamically robust means that the threshold is the smallest among all rules with this property.

Finally, our performance criterion can be used to quantify the value of information about the environment. The dynamically robust performance ratio bounds the ratio of payoffs of two individuals: an ignorant one (who knows nothing about the environment) and an informed one (who knows everything about the environment). Thus, it defines the maximal payoff loss due to being uninformed about the environment.

\subsection{First Insights}\label{s:insights}

Before unveiling our results, we present three simple, but important insights.

\subsubsection{Irrelevance of Priors.} 

The greatest obstacle in Bayesian optimization is that the problem of finding an optimal rule is generally intractable and only solvable for extremely simple priors. Our approach does not have this drawback, as we do not need to consider general priors. 
Below we show that it is enough to restrict attention to pure environments. 

Note that optimal rules under pure environments are simple to find, as these are cutoff rules that require to search until a certain cutoff is exceeded. Specifically, by \citeasnoun{Weitzman79}, the optimal rule under any given environment $F$ prescribes to stop whenever the best-so-far alternative $y$ exceeds a {\it reservation value} $c_F$ given by
\begin{equation}\label{E:RV}
c_F=\delta\left(\int_{0}^{c_F} c_F \df F(x) +\int_{c_F}^\infty x\df F(x)\right).
\end{equation}
The optimal payoff, given a best-so-far alternative $y$ and an environment $F$, is
\begin{equation}\label{OptV}
V(F,y)=\max\big\{y,c_F\big\}.
\end{equation}

The proposition below shows that the performance ratio of a rule can be determined by looking only at the pure environments. Recall that $\mathcal F(h)$ and $\Delta(\mathcal F(h))$ denote the set of environments and priors, respectively, that are consistent with a history $h$.

\begin{proposition}\label{L:Mixed}
For each decision rule $p$ and each history $h$,
\[
\inf_{\mu\in \Delta(\mathcal F(h))} \frac{U_p(\mu,h)}{V(\mu,h)}=\inf_{F\in\mathcal F(h)}\frac{U_p(F,h)}{V(F,h)}.
\]
\end{proposition}
Note that the ratio $U_p(\mu,h)/V(\mu,h)$ is nonlinear in $\mu$, so the result does not immediately follow from the fact that each prior $\mu$ is a linear combination of points in $\mathcal F$. The proof is in Appendix \ref{s:p1}.

\subsubsection{Irrelevance of Histories.} 
How should the individual condition her decisions on past observations? For instance, what does the individual learn after having observed a history $(x_0, x_1,...,x_n)$? All environments are still possible, except for the degenerate ones that assign zero probability to the values of $x_1,...,x_n$. When the set of feasible environments is convex, exclusion of these degenerate environments does not change the  infimum of the payoff ratios. Intuitively, this is because our performance measure involves evaluating the payoff ratio for each environment under which a given history occurs with a positive probability. How likely this history occurs does not influence the payoff ratio. If the history contains observations that cannot be generated by some environment $F$, other environments arbitrarily close to $F$ can generate this history with a positive, albeit arbitrarily small probability, and $F$ is a limit of a sequence of such environments.

\begin{proposition}\label{P:Hist}
Let $\mathcal F$ be convex. For each decision rule $p$ and each history $h$,
\[
\inf_{F\in\mathcal F(h)} \frac{U_p(F,h)}{V(F,h)}=\inf_{F\in\mathcal F} \frac{U_p(F,h)}{V(F,h)}.
\]
\end{proposition}

Proposition \ref{P:Hist} states that, when evaluating the infimum of the payoff ratio, one should take into account the set of all environments, regardless of whether or not they are consistent with the observed history. The proof is in Appendix \ref{s:phist}.

\subsubsection{Necessity to Randomize.} \label{s-ntr}
We now show that dynamically robust rules necessarily involve randomization. Stopping with certainty in any round is bad, because one might miss out a high realization in the next round. Yet, continuing forever with certainty is bad too, because this destroys the value of the outside option. We show that no deterministic rule can guarantee a better performance ratio than the rule that stops in round zero.

Specifically, if one stops and obtains $x_0$, the maximum possible foregone payoff is $\sup_{F\in \mathcal F} V(F,x_0)$. Thus, a performance ratio of $x_0/(\sup_{F\in \mathcal F} V(F,x_0))$ is trivially obtained by stopping in round zero, that is, by not searching at all.   

A decision rule $p$ is called {\it deterministic} if $p(h)\in\{0,1\}$ for each history $h $. Let $F_0$ denote the Dirac environment that almost surely generates an alternative that has value 0. The next proposition shows that deterministic decision rules cannot perform better than not searching at all, as long as the environment $F_0$ is feasible.

\begin{proposition}\label{P:Det}
Let $p$ be a deterministic decision rule. Suppose that $F_0\in \mathcal F$. Then
\[
R_{p}(x_0,\mathcal F)\le \frac{x_0}{\sup_{F\in \mathcal F} V(F,x_0)}.
\]
\end{proposition}
In particular, if the set of alternatives $X$ is unbounded and all distributions are feasible, so $\mathcal F=\mathcal F_X$, then $R_{p}(x_0,\mathcal F)=0$. The proof is in Appendix \ref{s:det}. 

\begin{remark}
Proposition \ref{P:Det} sheds light on the performance of decision rules used by Bayesian decision makers. By definition, any such rule is optimal for some prior. It stops the search if the best-so-far alternative is better than the expected continuation payoff under this prior, and continues otherwise. So, it is generically deterministic.\footnote{Indifference between stopping and continuing under a given prior is nongeneric, in the sense that it does not hold under an open set of priors in the neighborhood of that prior.}$^,$\footnote{This genericity follows from our assumption that the distribution of alternatives is exogenous. In Janssen et al.~(\citeyear*{Janssen2017}) the distribution is endogenous, and the equilibrium Bayesian search rule is nondeterministic.} Thus, by Proposition \ref{P:Det}, for some priors, this rule is never better than not searching at all. 
\end{remark}

\section{Binary Environments}\label{s:binary}
Consider the simple case in which feasible environments can have at most one value above the outside option. We call such environments {\it binary}. This case is relevant for applications where the individual knows what she is looking for, she just does not know whether she will find it and, if so, how valuable it will be. 

An environment is called {\it binary}, denoted by $F_{(z,\sigma)}$, if it is a lottery over two values, $0$ and $z$, with probabilities $1-\sigma$ and $\sigma$, respectively. The assumption that the low alternative has value 0 is for convenience: the results do not change, as long as at most one alternative above the outside option realizes with positive probability. Even if the individual does not know the value of this alternative at the outset, she immediately knows it after it has realized, and stops the search. In particular, the assumption of free recall plays no role for these environments.

Given a set $X$ of feasible alternatives, we denote by ${\mathcal B}_X$ the set of all binary environments over $X$, so
\[
{\mathcal B}_X=\{F_{(z,\sigma)}: z\in X,\,\sigma\in[0,1]\}.
\] 
A special case of only two feasible alternatives, $X=\{0,z\}$, captures the situation where the individual knows the value of the high alternative. In this case, the only unknown parameter is how likely the high alternative emerges in each next round.

When facing a set ${\mathcal B}_X$ of binary environments, any decision rule is fully described by a sequence of probabilities
\[
{q}=(q_0,q_1,q_2,...),
\]
where $q_t$ is the probability to stop in round $t$ conditional on only alternative 0 realized in rounds $1,...,t$. 

A decision rule $\bar q$ is {\it stationary} if its stopping probability is constant, so $q_0=q_1=q_2=...$. We will show that a particular stationary rule is dynamically robust in binary environments. 

We now present our result for binary environments.

\begin{theorem}\label{T:HalfP}
The stationary decision rule with the stopping probability $\ul q^*=\frac{1-\delta}{2-\delta}$

(a) attains the performance ratio $1/2$;

(b) is dynamically robust if $\sup X=\infty$.
\end{theorem}

Theorem \ref{T:HalfP} shows that one can always guarantee at least $1/2$ of the optimal payoff against binary environments. Moreover, this bound is tight if the set of feasible alternatives is unbounded. 

We now sketch the argument why this bound is attained. The probability of stopping $\ul q^*=\frac{1-\delta}{2-\delta}$ balances the payoff ratio between environments where it is optimal to stop and where it is optimal to keep searching until the high alternative realizes. Let us fix an outside option $x_0$. To simplify notation, we write $v^*_{(z,\sigma)}$ for the optimal payoff in a binary environment $F_{(z,\sigma)}$, and $u^*_{(z,\sigma)}$ for the individual's payoff from the rule that stops with constant probability $\ul q^*$ in that environment.
Observe that 
\begin{equation*}
u^*_{(z,\sigma)}=\ul q^* x_0+(1-\ul q^*)\delta (\sigma \max\{z,x_0\}+(1-\sigma)u^*_{(z,\sigma)}).
\end{equation*}
Substituting  $\ul q^*=\frac{1-\delta}{2-\delta}$ and solving for $u^*_{(z,\sigma)}$ yields
\[
u^*_{(z,\sigma)}=\frac{(1-\delta)x_0+\delta\sigma  \max\{z,x_0\}}{2(1-\delta)+\delta \sigma}.
\]
First, consider an environment where it is optimal to stop immediately, so $v^*_{(z,\sigma)}=x_0$. The payoff ratio is
\begin{equation}\label{e:bin-expl1}
\frac{u^*_{(z,\sigma)}}{v^*_{(z,\sigma)}}=\frac{(1-\delta)x_0+\delta\sigma  \max\{z,x_0\}}{2(1-\delta)+\delta \sigma}\cdot\frac 1 {x_0}\ge \frac{(1-\delta)x_0+\delta\sigma  x_0}{(2(1-\delta)+\delta \sigma)x_0}\ge \frac 1 2,
\end{equation}
where the first inequality by $\max\{z,x_0\}\ge x_0$, and the second inequality is because the ratio is increasing in $\sigma\in[0,1]$.

Second, consider an environment where $z>x_0$ and, moreover, it is optimal to search until $z$ realizes, so the optimal payoff is
\[
v^*_{(z,\sigma)}=\delta(\sigma z+(1-\sigma)v^*_{(z,\sigma)}).
\]
Solving for $v^*_{(z,\sigma)}$ yields $v^*_{(z,\sigma)}=\delta\sigma z/(1-\delta+\delta\sigma)$. The payoff ratio is
\begin{equation}\label{e:bin-expl2}
\frac{u^*_{(z,\sigma)}}{v^*_{(z,\sigma)}}=\frac{(1-\delta)x_0+\delta\sigma  z}{2(1-\delta)+\delta \sigma}\cdot\frac{1-\delta+\delta\sigma}{\delta\sigma z}\ge \frac{\delta\sigma  z}{2(1-\delta)+\delta \sigma}\cdot\frac{1-\delta+\delta\sigma}{\delta\sigma z}\ge \frac 1 2,
\end{equation}
where the first inequality is by $(1-\delta)x_0\ge 0$ and the second inequality is because the ratio is increasing $\sigma\in[0,1]$. Finally, notice that inequality \eqref{e:bin-expl1} holds as equality when $\sigma=0$ and would be violated for any stopping probability smaller than $\ul q^*$; and inequality \eqref{e:bin-expl2} holds as equality when $\sigma z\to\infty$ and would be violated for any stopping probability greater than $\ul q^*$. So, the performance ratio cannot be improved upon when $\sup X=\infty$. The formal proof is in Section \ref{s-s}. 

We achieve a better performance when environments are bounded. For each $x\in[0,1]$ define
\begin{equation}\label{E:PS}
q^*(x)=\frac{2(1-\delta)}{4-2\delta+ x-\sqrt{ x\left( x+8\right)}}\quad\text{and}\quad
\rho\left( x\right)=\frac 1 2+\frac 1 8\left( x+\sqrt{ x\left( x+8\right)}\right).
\end{equation}

\renewcommand{\thetheorem}{\arabic{theorem}$'$}
\addtocounter{theorem}{-1}
\begin{theorem}\label{T:Half}
Let $\bar x=\sup X<\infty$ and let $0<x_0\le \bar x$. The stationary decision rule with the stopping probability $q^*(x_0/\bar x)$

(a) attains the performance ratio $\rho(x_0/\bar x)>1/2$;

(b) is dynamically robust if $x_0/\bar x \le \delta^2/(2-\delta)$.
\end{theorem}
\renewcommand{\thetheorem}{\arabic{theorem}}

\begin{remark} If $x_0/\bar x >\delta^2/(2-\delta)$, then the rule $q^*(x_0/\bar x)$  it is not dynamically robust (so, a higher performance ratio can be attained). Yet $q^*(x_0/\bar x)$ attains the performance ratio $\rho(x_0/\bar x)$ which is already very good in this case:
\[
\rho(x_0/\bar x)>\rho\left(\frac{\delta^2}{2-\delta}\right)=\frac 1 {2-\delta}>\delta \ \ \text{for all}  \ x_0/\bar x>\frac{\delta^2}{2-\delta}.
\]
The dynamically robust rule and its performance ratio for all $x_0/\bar x\in (0,1]$ are derived in Section \ref{s-s} below (see \eqref{Rule:Binary} and \eqref{E:R-Bin}). 
\end{remark}

Theorem \ref{T:Half} shows that one can guarantee more than $1/2$ if the set of alternatives is bounded. How much more one can guarantee depends on how large the outside option $x_0$ is relative to the highest feasible alternative $\bar x$. In fact, if $x_0$ is extremely small, then the performance ratio is close to $1/2$. Yet one can guarantee at least $2/3$ and $3/4$ of the optimal payoff if $x_0/\bar x$ exceeds, respectively, $1/6$ and $1/3$. Table 1 illustrates the performance ratio $\rho(x_0/\bar x)$ for a few values of $x_0/\bar x$.

\begin{table}[!htb]
\[
\begin{array}{r|cccccccc}
\text{$x_0/\bar x$} & 1/89 & 1/20 & 1/10 & 1/6 & 1/5 & 1/4 & 1/3 & 1/2 \\
\hline
\text{$\rho(x_0/\bar x)$} & 0.538 & 0.552 & 0.625 & 0.666 & 0.685 & 0.71 & 0.75 & 0.82
\end{array}
\]
\caption{\small Some values of the performance ratio of rule ${q}^*$}\label{T:1}
\end{table}


\begin{figure}[!h]
\begin{center}
\includegraphics[width=0.5\textwidth]{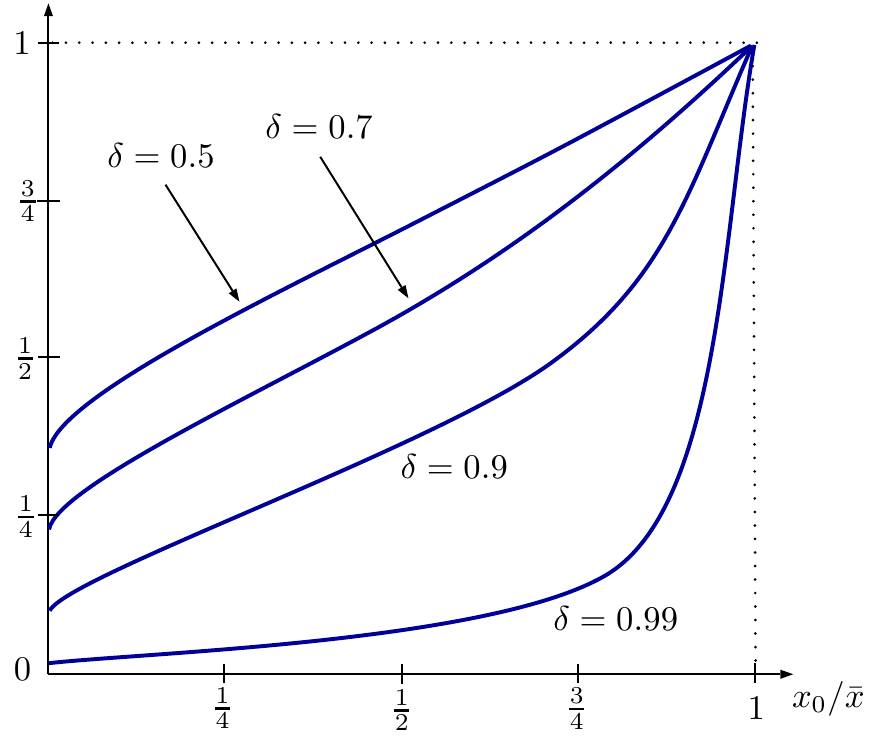}
\caption{\small Stopping probability $q^*(x_0/\bar x)$ with values of the discount factor $\delta=0.5$, $0.7$, $0.9$, and $0.99$.}
\label{F:0}
\end{center}
\end{figure}

Notice that the stopping probability $q^*(x_0/\bar x)$ of the stationary rule is increasing in $x_0/\bar x$. Figure~\ref{F:0} illustrates this stopping probability for some values of the discount factor.

Curiously, the performance ratios identified in Theorems \ref{T:HalfP} and \ref{T:Half} are independent of the discount factor $\delta$. Intuitively, this is because both the individual's payoff from the rule given by \eqref{E:PS} and the optimal payoff $V$ are evaluated using the same discount factor. So, when following a dynamically robust rule, a more patient individual simply waits longer in expectation. 

\subsection{Proof of Theorems \ref{T:HalfP} and \ref{T:Half}.} \label{s-s}
Part (a) in each of the theorems can be proven by a simple verification that the specified stationary decision rule yields at least the claimed performance ratio. The proof of part (b) is more involved as it requires to show that there is no other rule that attains a higher performance ratio. 

Specifically, there are two main stepping stones to the proof of Theorems \ref{T:HalfP} and \ref{T:Half}. The first stepping stone was provided in  Proposition \ref{L:Mixed}, where we showed that the performance ratio can be determined by looking only at the pure environments, so we do not need to worry about mixed environments and priors. 

The second stepping stone, which we now establish, is that we can restrict attention to stationary decision rules without loss of generality. This simplifies the problem tremendously, as any stationary rule is described by a single parameter: the constant stopping probability. So, it becomes a single-variable optimization problem.

Clearly, the individual should stop the search after observing any alternative other than zero, as she then knows that such an alternative is the best possible. However, as long as only zero-valued alternatives have realized, this history is irrelevant for the evaluation of the performance ratio, as shown by Proposition \ref{P:Hist}. That is, the individual faces the exact same problem again and again, as long as she draws zero-valued alternatives. Given this unchanging problem, there are neither fundamental nor strategic reasons to condition decisions on the history. We now show that this intuition is correct, so we can search for a dynamically robust rule among stationary rules.

\begin{proposition}\label{P:StatB}
For each decision rule ${q}$ there exists a stationary decision rule $\bar {q}$ such that $R_{\bar{q}}(x_0,\mathcal B_X)\ge R_{q}(x_0,\mathcal B_X)$.
\end{proposition}
The proof is in Appendix \ref{s:statb}.

We are now ready to prove Theorem \ref{T:Half}. Theorem \ref{T:HalfP} will follow by taking the limit of $\bar x\to\infty$ for a fixed $x_0$, so that $x_0/\bar x\to 0$.

Let $X$ be a set of feasible alternatives with $\bar x=\sup X<\infty$, and let $x_0\in(0,\bar x]$ be an outside option.
By Proposition \ref{L:Mixed}, we restrict attention to the set of pure environments, $\mathcal B_X$. By Proposition \ref{P:StatB}, we consider stationary decision rules. Any such rule is identified with its constant probability of stopping, $q\in[0,1]$. 

For each binary environment $F_{(z,\sigma)}\in {\mathcal B}_X$, let $U_q(F_{(z,\sigma)},x_0)$ be the individual's expected payoff from a stationary rule $q$:
\begin{equation}\label{e-U0}
U_q(F_{(z,\sigma)},x_0)=q x_0+(1-q)\delta ((1-\sigma)U_q(F_{(z,\sigma)},x_0)+\sigma z).
\end{equation}
By \eqref{E:RV}, the reservation value $c_{F_{(z,\sigma)}}$ satisfies $c_{F_{(z,\sigma)}}=\delta(\sigma z+(1-\sigma)c_{F_{(z,\sigma)}})$. By  \eqref{OptV}, the optimal payoff is given by $V(F_{(z,\sigma)},x_0)=\max\big\{x_0,c_{F_{(z,\sigma)}}\big\}$. We thus obtain
\begin{equation}\label{E:RVS}
c_{F_{(z,\sigma)}}=\frac{\delta\sigma z}{1-\delta(1-\sigma)} \quad \text{and} \quad V(F_{(z,\sigma)},y)=\max\left\{y,\frac{\delta\sigma z}{1-\delta(1-\sigma)}\right\}.
\end{equation}
The performance ratio of rule $q$ is
\begin{align}
R_q(x_0, \mathcal B_X)&=\inf_{F\in \mathcal B_X} \frac{U_q(F_{(z,\sigma)},x_0)}{V(F_{(z,\sigma)},x_0)}=\inf_{F\in \mathcal B_X} \min\left\{\frac{U_q(F_{(z,\sigma)},x_0)}{x_0},\frac{U_q(F_{(z,\sigma)},x_0)}{c_{F_{(z,\sigma)}}}\right\}. \notag\\
&=\min\left\{\inf_{F\in \mathcal B_X} \frac{U_q(F_{(z,\sigma)},x_0)}{x_0},\inf_{F\in \mathcal B_X} \frac{U_q(F_{(z,\sigma)},x_0)}{c_{F_{(z,\sigma)}}}\right\}.\label{e-R0}
\end{align}
In words, the individual worries about two scenarios: $c_{F_{(z,\sigma)}}<x_0$, in which case it is optimal to stop immediately, and $c_{F_{(z,\sigma)}}\ge x_0$, in which case a high value of $z$ is sufficiently likely, and it is optimal to wait for it. The optimal stopping probability $q$ should be large in the first scenario and small in the second scenario, thus it should balance this tradeoff.
To find the optimal $q$, we evaluate the worst-case ratios for each of the two scenarios. 

Consider the first expression under the minimum in \eqref{e-R0}. Solving \eqref{e-U0} for $U_q(F_{(z,\sigma)},x_0)$ yields
\begin{equation}\label{e-U}
U_q(F_{(z,\sigma)},x_0)=\frac{q x_0+(1-q)\delta \sigma \max\{z,x_0\}}{1-\delta(1-\sigma)(1-q)}.
\end{equation}
We thus have
\begin{align}
\inf_{F_{(z,\sigma)}\in \mathcal B_X} \frac{U_q(F_{(z,\sigma)},x_0)}{x_0}=\inf_{z\in X,\sigma\in[0,1]} \frac{q x_0+(1-q)\delta \sigma \max\{z,x_0\}}{(1-\delta(1-\sigma)(1-q))x_0}= \frac{q}{1-\delta(1-q)},\label{e-r0}
\end{align}
where the last equality is because $\max\{z,x_0\}\ge x_0$ and the ratio is increasing in $\sigma$, and thus achieves the minimum at $\sigma=0$. So, the worst-case environments in the first scenario are those environments $F_{(z,\sigma)}$ in which $\sigma=0$, so the high alternative $z$ never occurs.

Next, consider the second expression under the minimum in \eqref{e-R0}. 
By \eqref{E:RVS} and \eqref{e-U},
\begin{align}
\inf_{F_{(z,\sigma)}\in \mathcal B_X}\frac{U_q(F_{(z,\sigma)},x_0)}{c_{F_{(z,\sigma)}}}&=\inf_{z\in X,\sigma\in[0,1]} \frac{q x_0+(1-q)\delta \sigma z}{1-\delta(1-\sigma)(1-q)}\cdot\frac{1-\delta(1-\sigma)}{\delta\sigma z}\notag\\
&= \inf_{\sigma\in[0,1]}\left( \inf_{z\in X}\frac{q \frac{x_0}{z}+(1-q)\delta \sigma}{(1-\delta(1-\sigma)(1-q))\frac{\delta\sigma}{1-\delta(1-\sigma)}}\right)\notag\\
&= \inf_{\sigma\in[0,1]}\frac{q \frac{x_0}{\bar x}+(1-q)\delta \sigma}{(1-\delta(1-\sigma)(1-q))\frac{\delta\sigma}{1-\delta(1-\sigma)}},\label{e-r1}
\end{align}
where the last equality is by $\inf_{z\in X} x_0/z=x_0/\bar x$. So, worst-case environments in the second scenario are those environments $F_{(z,\sigma)}$ in which $z=\bar x$, so $z$ is the highest possible alternative.

Thus, from \eqref{e-r0} and \eqref{e-r1}, we need to solve
\begin{align}
\max_{q\in[0,1]}\min_{\sigma\in[0,1]}\left(\min\left\{\frac{q}{1-\delta(1-q)},\frac{q \frac{x_0}{\bar x}+(1-q)\delta \sigma}{(1-\delta(1-\sigma)(1-q))\frac{\delta\sigma}{1-\delta(1-\sigma)}}\right\}\right).\label{e-R30}
\end{align}
Denote $\hat x=x_0/\bar x$. It is straightforward to verify that the unique solution $(\bar q,\bar \sigma)$ of the maximin problem \eqref{e-R30} is
\begin{align}
\bar q&=\begin{cases}
\frac{2(1-\delta)}{%
4-2\delta+\hat x-\sqrt{\hat x\left( \hat x+8\right) }}, & 
\text{if} \ 0<\hat x\leq \frac{\delta ^{2}}{2-\delta }, \\ 
\frac{\sqrt{(1-\delta)((2\delta-\hat x)^2-\delta \hat x^2)}-(1-\delta)(2\delta-\hat x)}{2\delta(\delta-\hat x)}, & \text{if} \ \frac{%
\delta ^{2}}{2-\delta }<\hat x<\delta,\\
1, & \text{if} \ \delta\le \hat x\le 1,
\end{cases}\label{Rule:Binary}\\
\bar \sigma&=\begin{cases}
\frac{(1-\delta)(3\hat x+\sqrt{\hat x\left( \hat x+8\right) })}{2\delta \left( 1-\hat x\right) }, &\text{if} \ 0<\hat x\leq \frac{\delta ^{2}}{2-\delta },\\
1,& \text{if} \ \frac{\delta ^{2}}{2-\delta }<\hat x\le 1.
\end{cases}\notag
\end{align}
We thus derived a dynamically robust decision rule $\bar q$. Substituting $\bar q$ into \eqref{e-R0} yields the dynamically robust performance ratio
\begin{equation}\label{E:R-Bin}
R^*(x_0, \mathcal B_X)=\begin{cases}
\frac 1 2+\frac 1 8\left(\hat x+\sqrt{\hat x\left(\hat x+8\right)}\right), & \text{if} \ 0<\hat x\leq \frac{\delta ^{2}}{2-\delta }, \\ 
\frac{2\delta-(1-\delta)\hat x-\sqrt{(1-\delta)((2\delta-\hat x)^2-\delta \hat x^2)}}{2\delta^2}, & \text{if} \ \frac{\delta ^{2}}{2-\delta }<\hat x<\delta,\\
1, & \text{if} \ \delta\le \hat x\le 1.
\end{cases}
\end{equation}
Finally, observe that decision rule $q^*(\hat x)$ given by \eqref{E:PS} coincides with the dynamically robust rule $\bar q$ for $\hat x=\frac{x_0}{\bar x}\le \frac{\delta ^{2}}{2-\delta }$. For all $0<\hat x\le 1$ it yields the performance ratio
\[
R_{q^*}(x_0, \mathcal B_X)=\frac{q^*(\hat x)}{1-\delta(1-q^*(\hat x))}=\frac 1 2+\frac 1 8\left(\hat x+\sqrt{\hat x\left(\hat x+8\right)}\right)=\rho(\hat x)> \frac 1 2.
\]
This completes the proof of Theorem \ref{T:Half}.

\section{General Environments}

Consider now more general environments that potentially generate multiple alternatives above the outside option. To keep the exposition simple, we fix a set of alternatives $X$ and allow for all distributions over $X$ that have finite support. So $\mathcal F=\mathcal F_X$. 

In contrast to binary environments, here the first alternative above the outside option need not be the best, so sometimes the individual may wish to search for even better alternatives. This makes decision making more complex. We deal with this complexity by building on and extending our insights obtained for the binary setting. Once again, we can restrict attention to simple decision rules, which here means that they are stationary and have some monotonicity properties. We can also restrict attention to binary environments, as only these determine the worst-case payoff ratio.

\subsection{Simplicity of Decision Rules}\label{s-stat}
By Proposition \ref{P:Hist}, histories are irrelevant for the evaluation of the performance ratio. The only payoff-relevant variable is the best-so-far alternative. Intuitively, the individual has no reason to condition decisions on anything other than the best-so-far alternative. This suggests that we can restrict attention to stationary decision rules, in which the probability of stopping in each round depends only on the best-so-far alternative. 

Formally, a decision rule $p$ is {\it stationary} if the stopping probability is the same for any pair histories $h'$ and $h''$ with same best-so-far alternative, so
\[
\max \{x:x\in h'\}=\max \{x:x\in h''\} \implies p(h')=p(h'') \quad \text{for all $h',h''\in\mathcal H(x_0)$.}
\]
With stationary decision rules, we simplify notation by replacing each history $h_t$ with the best-so-far alternative $y=\max\{x_0,x_1,...,x_t\}$. So, a stationary rule $p:\R_+\to [0,1]$ prescribes for each best-so-far alternative $y$ to stop with probability $p(y)$. For each environment $F$ and each best-so-far alternative $y$, the optimal payoff is given by
\[
V(F,y)=\max_{q\in[0,1]} \left(q y+(1-q)\delta \int_0^\infty V(F,\max\{y,x\})\df F(x)\right),
\]
and the payoff of rule $p$ is given by
\[
U_p(F,y)=p(y) y+(1-p(y))\delta \int_0^\infty U_p(F,\max\{y,x\})\df F(x).
\]

We introduce two intuitive properties of a stationary decision rule.

A stationary decision rule $p$ is {\it monotone} if $p(y)$ is weakly increasing. It is natural that the individual is more likely to accept a greater best-so-far alternative.

A stationary decision rule $p$ has the {\it monotone ratio property} if 
\begin{equation}\label{E:r-p}
r_{p}(y):=\inf_{F\in {\mathcal B}_X}\frac{U_p(F,y)}{V(F,y)} \ \ \text{is weakly increasing}.
\end{equation}
This is a ``free-disposal'' property. Suppose that the best-so-far alternative has increased from $y$ to $y'$, but the payoff ratio has decreased. Then the individual could be better off by destroying some part of the value of the best-so-far alternative and decreasing it back to $y$. 

Decision rules in general environments can be very complex. The next proposition shows we can restrict attention to much simpler decision rules, namely, those that are stationary, monotone, and have the monotone ratio property.

\begin{proposition}\label{P:Stat}
For each decision rule $p$, there exists a stationary monotone decision rule $\tilde p$ with the monotone ratio property such that $R_{\tilde p}(x_0,\mathcal F_X)\ge R_{p}(x_0,\mathcal F_X)$.
\end{proposition}

The proof is in Appendix \ref{s:pstat}.

\subsection{Simplicity of Worst-Case Environments}\label{s-wce}
Proposition \ref{P:Stat} shows that there exist simple dynamically robust rules. We now show that their simple nature causes worst-case environments to be very simple, too. Specifically, these environments are binary. 

\begin{proposition}\label{P:Binary} Let decision rule $p$ be stationary, monotone, and satisfy the monotone ratio property. Then
\[
R_p(x_0,\mathcal F_X)=\inf_{y\ge x_0}\inf_{F\in {\mathcal B}_X} \frac{U_p(F,y)}{V(F,y)}.
\]
\end{proposition}

The proof is in Appendix \ref{s:pbin}.

To gain the intuition for Proposition \ref{P:Binary}, recall that the individual cares about two contingencies: stopping when she should have waited for a higher realization of the value, and continuing when there are no better alternatives in the future. The worst-case distributions for these contingencies need not be complex, they are binary valued. 

We hasten to point out that Proposition \ref{P:Binary} does not imply that the individual should act as if she faces binary environments, as otherwise she would stop after seeing any alternative above the outside option. Instead, Proposition \ref{P:Binary} implies that, when evaluating the payoff ratio after any history of realized alternatives, we only need to do so for {\it all} binary environments. The value of Proposition \ref{P:Binary} is that it drastically simplifies the calculation of the performance ratio. 

Note that binary environments are not consistent with histories that contain more than two values. However, by Proposition \ref{P:Hist}, we should not be worried about this inconsistency, as any binary distribution that is inconsistent with a history can be obtained as a limit of a sequence of distributions that are consistent with that history.

\subsection{Dynamically Robust Performance}
We are now ready to present our findings for general environments.
\begin{theorem}\label{T:Quarter}
The stationary decision rule $\bar p$ given for each $y$ by
\[
\bar p(y)=\frac{1-\delta}{2-\delta}
\]

(a) attains the performance ratio $R_{\bar p}(x_0,\mathcal F_X)\ge 1/4$;

(b) is dynamically robust if $\sup X=\infty$.
%
\end{theorem}

The proof is in Appendix \ref{s:tquarter}.

Part (a) shows that the dynamically robust performance ratio $R^*(x_0,\mathcal F_X)$ is at least $1/4$ against general environments. Part (b) shows that this bound is tight as it is attained when the set of feasible alternatives is unbounded. 

We now sketch the argument why this bound is attained. For any  best-so-far alternative $y$, the relevant worst-case environments are those where there is an alternative $z$ that is very large relative to $y$. The stopping probability $\frac{1-\delta}{2-\delta}$ balances the payoff ratio between environments where $z$ is sufficiently unlikely (so it is optimal to stop) and environments where $z$ is likely enough and is worth waiting for.
Consider decision rules with a constant stopping probability, $q$, in each round, in particular, before and after $z$ realizes. A greater $q$ means a shorter delay of obtaining $z$ after it has realized, but also a greater probability of stopping before the first realization of $z$. In the limit, as $y/z$ tends to 0, the payoff ratio takes the form
\begin{equation}\label{e:PP}
\frac{q}{1-\delta (1-q)}\left(1-\frac{q}{1-\delta (1-q)}\right).
\end{equation}
The first factor in \eqref{e:PP}, 
\[
\frac{q}{1-\delta (1-q)}=q+\delta(1-q)q+\delta^2(1-q)^2q+... 
\]
is a reciprocal of the expected delay of obtaining $z$ after its realization. The second factor in \eqref{e:PP} is the probability of not stopping before $z$ realizes for the first time. Setting $\frac{q}{1-\delta (1-q)}$ equal to $1/2$ maximizes \eqref{e:PP}, leading to the solution $q=\frac{1-\delta}{2-\delta}$ and the guaranteed payoff ratio $1/4$.

Analogously to the binary setting, we achieve a better performance when feasible alternatives are bounded. 
For this result, recall the definition of $q^*$ and $\rho$ given by \eqref{E:PS} in Section \ref{s:binary}.

\renewcommand{\thetheorem}{\arabic{theorem}$'$}
\addtocounter{theorem}{-1}
\begin{theorem}\label{P:Rules}
Let $\bar x=\sup X<\infty$. Then there exists a constant $L\in(0,1)$ such that the dynamically robust performance ratio satisfies
\[
R^*(x_0,\mathcal F_X)\ge \rho(x_0/\bar x)>1/2 \ \ \text{if $x_0/\bar x\ge L$.}
\]

Moreover, if $x_0/\bar x \ge 1/6$, then the decision rule $p^*$ given by $p^*(y)=q^*(y/\bar x)$

(a) attains the performance ratio $\rho(x_0/\bar x)>1/2$;

(b) is dynamically robust if\,\footnote{If $x_0/\bar x$ is not in $[1/6, \delta^2/(2-\delta)]$, then the rule $p^*$  it is not dynamically robust. The dynamically robust rule and its performance ratio for each $x_0/\bar x\in [L,1]$ are derived in the proof of Theorem \ref{P:Rules} (Appendix \ref{s:prules}).} $1/6\le x_0/\bar x \le \delta^2/(2-\delta)$.
\end{theorem}
\renewcommand{\thetheorem}{\arabic{theorem}}

The proof is in Appendix \ref{s:prules}.

Theorem \ref{P:Rules} shows that, if the outside option is not too small relative to the highest possible alternative, in the sense that $x_0/\bar x\ge L$, then the dynamically robust performance ratio for the general environments is the same as that for the binary environments. That is, the expansion from the binary to general set of environments confers no reduction in the dynamically robust performance. Remarkably, the constant $L$ is very small. We numerically find an upper bound for $L$:
\[
L\le1/89, 
\]
which is independent of the discount factor $\delta$. Thus, as $x_0/\bar x$ increases from $0$ to a mere $1/89$, the dynamically robust performance ratio climbs from $1/4$ to above $1/2$. In particular, one can guarantee at least $2/3$ and $3/4$ of the optimum if the outside option exceeds, respectively, $1/6$ and $1/3$ of the highest possible alternative.  In Figure~\ref{F:1}, the dynamically robust performance ratio $R^*(x_0,\mathcal F_{X})$ is shown as a solid line for $x_0/\bar x\ge L$, and we hypothesize that it looks as depicted by the dotted line for $x_0/\bar x<L$.

\begin{figure}[!h]
\begin{center}
\includegraphics[width=0.5\textwidth]{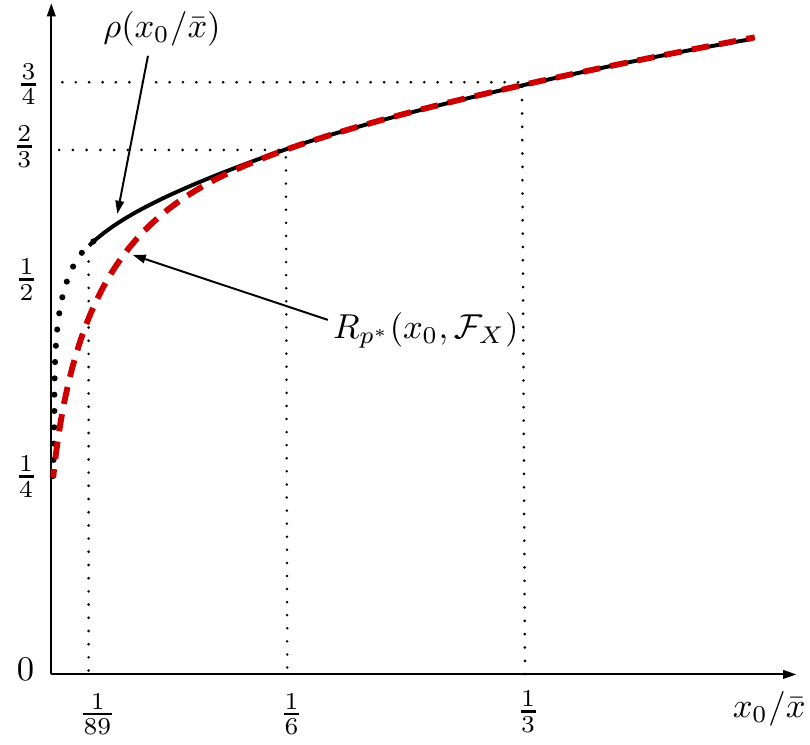}
\caption{\small Dashed line shows the performance ratio of rule $p^*$. Solid line shows the dynamically robust performance ratio $\rho(x_0/\bar x)$ for $x_0/\bar x\ge L$ when $\delta$ is sufficiently large, so $x_0/\bar x\le \delta^2/(2-\delta)$. Dotted line shows the hypothetical value of the dynamically robust performance ratio when $x_0/\bar x<L$.}
\label{F:1}
\end{center}
\end{figure}

In addition, Theorem \ref{P:Rules} shows that the dynamically robust rule identified for binary environments is also dynamically robust in general environments when $x_0/\bar x\ge 1/6$. This can be seen visually in Figure~\ref{F:1} by the fact that the dynamically robust performance ratio (solid line) coincides with the performance ratio of the rule $p^*$ (dashed line) when $x_0/\bar x\ge 1/6$.

For $x_0/\bar x\in[L,1/6)$, the rule $p^*$ is no longer dynamically robust. In Figure~\ref{F:1}, the dashed line showing $R_{p^*}(x_0, \mathcal F_X)$ is below the solid line showing $\rho(x_0)$. Nevertheless, the performance ratio of $\rho(x_0)$ can still be attained. In our proof of Theorem \ref{P:Rules} we present a rule that is dynamically robust in this case.

Two elements of Theorem \ref{P:Rules} prompt curiosity. 

First, why is the dynamically robust performance ratio the same as under binary environments for such a large interval of outside options? It turns out that when  $x_0/\bar x\ge L$, the environments that determine the worst-case ratio are lotteries over the extreme alternatives, $0$ and $\bar x$. Other alternatives do not play any role in this parameter region, but they do when $x_0/\bar x<L$. This stands in contrast to the setting of binary environments where the worst case environments only put weights on the extreme alternatives.

Second, why is Theorem \ref{P:Rules} silent about dynamic robustness when $0<x_0/\bar x<L$? 
This is because a closed form expression for the dynamically robust performance ratio is not available for this region. Yet, for each parameter value in this region, the dynamically robust performance ratio, together with an associated decision rule, can be derived using a recursive procedure that we describe in Appendix \ref{s:Deriv}.


\section{Conclusion}

It is difficult to search when the distribution of alternatives is not known. In fact, as outlined in the literature review in the introduction, the literature has not produced satisfactory insights into how to search in this setting. In this paper we identify that this difficulty is due to the desire to achieve the very highest payoff for the given beliefs. Namely, we find that it is easier to search if one reduces the target and replaces ``very highest'' by ``relatively high''. The ease refers to the ability to derive an optimal solution for a very general setting, the simplicity of our algorithm, and the minimality of assumptions one needs to impose on the environment. 

The methodology developed in this paper is general, applicable to a spectrum of dynamic decision making problems, and should spark future research. Its strength is that it allows for dynamically consistent decision making with multiple priors. It is as if our decision maker is surrounded by other individuals, each of whom has her own prior. At any point in time, each of these individuals wants to complain that our rule is not appropriate given their prior. According to our concept, a dynamically robust rule may not be optimal given their prior, but their complaints cannot be large. Our analysis reveals bounds on the size of any such complaint. 

Our results about the bounds on the size of possible complaints do not change if the searcher has more information about the environment, for example, if she can restrict the set of priors, or if the number of alternatives is finite and known. What does change is the tightness of these bounds. Of course, when more information is available, better decision rules can be found. 

The main insights (randomization is essential, dynamically robust rules are stationary, worst-case priors are simple) extend to search without recall and to search with exchangeable distributions. 

New avenues for research on dynamical robustness are opened, such as extending this agenda to matching and to other search environments that involve strategic interaction. The economic insights of models that include agents searching under known distributions can now be reevaluated using agents that employ dynamically robust search.

\section*{Appendix A. Binary Environments}
\renewcommand{\thesection}{A}

\subsection{Proof of Proposition \ref{L:Mixed}.}\label{s:p1}
Fix a history $h$ and a prior $\mu$ that is consistent with that history, so $\mu\in\Delta(\mathcal F(h))$. Let
\[
r_p( F)=\frac{U_p( F,h)}{V( F,h)} \quad\text{and}\quad \eta ( F)=\frac{V( F,h)\mu( F|{h})}{V(\mu,h)}.
\]
Note that the posterior $\mu(\cdot|h)$ must assign zero probability to the set of all environments that are inconsistent with $h$, thus
\[
U_p(\mu,h)=\sum_{F\in\mathcal F(h)} U_p( F,h) \mu( F|{h}).
\]
Using the above notations we obtain
\begin{align*}
\frac{U_p(\mu,h)}{V(\mu,h)}&=\frac{\sum_{ F\in\mathcal F(h)} U_p( F,h) \mu( F|{h})}{V(\mu,h)}=\frac{\sum_{ F\in\mathcal F(h)} r_p( F)V( F,h) \mu( F|{h})}{V(\mu,h)}\\
&=\sum_{ F\in\mathcal F(h)} r_p( F)\eta( F)\ge\inf_{ F\in\mathcal F(h)} r_p( F)=\inf_{ F\in\mathcal F(h)} \frac{U_p( F,h)}{V( F,h)},
\end{align*}
where the inequality follows from  $r_p( F)\ge 0$, $\eta(F)\ge 0$, and 
\begin{align*}
\sum_{ F\in \mathcal F}\eta( F)&=\frac{\sum_{ F\in \mathcal F}V(F,h)\mu(F|h)}{V(\mu,h)}=\frac{\sum_{ F\in \mathcal F}\sup_p U_p(F,h)\mu(F|h)}{V(\mu,h)}\\
&\ge \frac{\sup_p\sum_{ F\in \mathcal F} U_p(F,h)\mu(F|h)}{V(\mu,h)}=\frac{\sup_p U_p(\mu,h)}{V(\mu,h)}=\frac{V(\mu,h)}{V(\mu,h)}=1.
\end{align*}
Since the above holds for all $\mu\in\Delta(\mathcal F(h))$, we have 
\[
\inf_{\mu\in\Delta(\mathcal F(h))}\frac{U_p(\mu,h)}{V(\mu,h)}\ge \inf_{ F\in\mathcal F(h)}\frac{U_p( F,h)}{V( F,h)}.
\]
The proof of the reverse of the above inequality is trivial, since $\mathcal F(h)$ is a subset of $\Delta(\mathcal F(h))$.

\subsection{Proof of Proposition \ref{P:Hist}.}\label{s:phist}
Fix an outside option $x_0>0$ and a history $h\in\mathcal H(x_0)$. Recall that $\mathcal F(h)\subset \mathcal F$ denotes the set of environments that are consistent with a history $h$, and note that $\mathcal F(h)\neq\varnothing$.  Consider two environments, $F\in\mathcal F(h)$ and $G\in\mathcal F$. Let $(G_k)_{k=1}^\infty$ be a sequence of environments given by
\[
G_k=\tfrac 1 k F+ \left(1-\tfrac 1 k\right) G, \ \ k\in\N,
\]
so $\lim_{k\to\infty}G_k=G$. By convexity of $\mathcal F$, $G_k\in\mathcal F$ for all $k\in\N$. Consistency of $F\in\mathcal F(h)$ with history $h=(x_0,x_1,...,x_t)$ means that $supp(F)$ contains $\{x_1,...,x_t\}$. 
So, $\{x_1,...,x_t\}\subset supp(F)\subset supp(G_k)$, and thus $G_k\in\mathcal F(h_t)$ for all $k\in\N$. Since the above is true for all $G\in\mathcal F$, it follows that $Closure(\mathcal F(h))= \mathcal F$, which proves the proposition.

\subsection{Proof of Proposition \ref{P:Det}.}\label{s:det}
Let $p$ be deterministic. Suppose that there exists $k\in\{0,1,2,...\}$ such that $p$ stops searching after $k$ zero-valued alternatives. Formally, $p(x_0\oplus{\bf 0}^k)=1$, where ${\bf 0}^k$ denotes the sequence of $k$ zeros and `$\oplus$' denotes the vector concatenation operator. For any $F\in\mathcal F$, the individual's payoff in round $k$ is $U_{p}(F,x_0\oplus {\bf 0}^k)=x_0$, and the optimal payoff in round $k$ satisfies
\[
V(F,x_0\oplus{\bf 0}^k)\le \sup\nolimits_{F'\in\mathcal F} V(F',x_0), 
\]
because, by \eqref{OptV}, $V(F,x_0)=V(F,x_0\oplus{\bf 0}^k)$. Consequently,
\[
R_p(x_0,\mathcal F)\le\inf_{F\in\mathcal F}\frac{U_p(F,x_0\oplus{\bf 0}^k)}{V(F,x_0\oplus{\bf 0}^k)}= \frac{x_0}{\sup_{F\in\mathcal F} V(F,x_0)}.
\]

Now, consider the complementary case where $p$ never stops searching as long as only zeros occurred in the past. So, $p(x_0\oplus{\bf 0}^k)=0$ for all $k\in\{0,1,2,...\}$. Consider an environment $F_0$ in which all alternatives are equal to zero with certainty. In round 0, the optimal payoff under $F_{0}$ is $V(F_{0}, x_0)=x_0$. Since $p$ continues after each history with only zeros, it never stops under $F_{0}$, and hence its payoff is $U_p(F_{0}, x_0)=0$. Consequently,
\[
R_p(x_0,\mathcal F)\le \frac{U_p(F_{0}, x_0)}{V(F_{0}, x_0)}=\frac{0}{x_0}=0.
\]

\subsection{Proof of Proposition \ref{P:StatB}.}\label{s:statb}
In what follows, we denote by $\bar q^\infty$ a constant sequence, so $\bar q^\infty=(\bar q,\bar q,...)$ for $\bar q\in[0,1]$.

We show that each sequence of probabilities  ${q}'=( q'_0, q'_1,...)$ can be replaced by a constant sequence $\bar{q}^\infty$ that has a weakly higher performance ratio in binary environments. Note that we only need to compare the individual's payoffs $U_{ q'}$ and $U_{\bar{ q}^\infty}$, as the optimal payoff $V$ does not depend on the decision rule.

For consistency with notations in Appendix B, we use notation $y=x_0$. In the paper, $y$ denotes the current best-so-far alternative, and in binary environments this is always the outside option $x_0$. Also, note that in binary environments we only have to consider histories in which only zeros occur, and hence replace $h_t$ by the round number $t$.

The expected payoff of a rule $q$ in each round $t=0,1,2,...$ is given by
\begin{equation}\label{E:Phi}
U_{ q}(F_{(z,\sigma)},t)= q_t y+(1- q_t)\delta(\sigma z+(1-\sigma)U_{ q}(F_{(z,\sigma)},t+1)).
\end{equation}
For each $(z,\sigma)$, denote the worst expected payoff among all rounds by
\[
\ul U_{ q}(F_{(z,\sigma)})=\inf_{t=0,1,...}U_{ q}(F_{(z,\sigma)},t).
\]
Let ${ q}'$ be an arbitrary sequence of probabilities. This ${ q}'$ will be called a {\it benchmark} and will be fixed for the rest of the proof. We say that a sequence ${ q}$ is {\it better than ${ q'}$ for $(z,\sigma)$} if its worst expected payoff under environment $F_{(z,\sigma)}$ is at least as good as that of the benchmark ${ q}'$, so
\[
\ul U_{{ q}}(F_{(z,\sigma)})\ge \ul U_{ q'}(F_{(z,\sigma)}).
\]
Let $\bar q^\infty=(\bar q,\bar  q,...)$ be the constant sequence where $\bar q$ is a solution of the equation
\begin{equation}\label{E:Benchmark}
U_{ q'}(F_{(z,0)},0)=\bar  q y+(1-\bar  q)\delta U_{ q'}(F_{(z,0)},0),
\end{equation}
so
\[
\bar q=\frac{(1-\delta)U_{ q'}(F_{(z,0)},0)}{y-\delta U_{ q'}(F_{(z,0)},0)}.
\]
By \eqref{E:Phi}, $U_{ q'}(F_{(z,0)},0)\in[0,y]$ when $\sigma=0$, so $\bar q\in[0,1]$. We will show that $\bar{q}^\infty$ is better than ${ q'}$ for all $z\ge y$ and all $\sigma\in[0,1]$.

By \eqref{E:Phi}, observe that for any sequence ${ q}$ and any $t$,
\begin{align*}
U_{{ q}}(F_{(z,\sigma)},t)-y&=(1- q_t)(\delta\sigma z+\delta(1-\sigma)U_{{ q}}(F_{(z,\sigma)},t+1)-y)\\
&=(1- q_t)(\delta\sigma z-(1-\delta(1-\sigma))y+\delta(1-\sigma) (U_{{ q}}(F_{(z,\sigma)},t+1)-y)).\end{align*}
Iterating the above for $t+1, t+2,...$, we obtain
\begin{equation}\label{E:A-TR1}
U_{{ q}}(F_{(z,\sigma)},t)-y=(\delta\sigma z-(1-\delta(1-\sigma))y)\sum_{k=0}^\infty \left(\delta^k(1-\sigma)^k\prod_{s=t}^{t+k}(1-q_s)\right).
\end{equation}

First, assume that $\delta\sigma z-(1-\delta(1-\sigma))y=0$. Then $U_{{ q}}(F_{(z,\sigma)},t)-y=0$ for every ${ q}$ and every $t$.  In particular, $\ul U_{\bar { q}^\infty}(F_{(z,\sigma)})=\ul U_{{ q}'}(F_{(z,\sigma)})=y$. So, $\bar { q}^\infty$ is better than ${ q'}$.

Next, assume that $\delta\sigma z-(1-\delta(1-\sigma))y\ne 0$. Define 
\begin{equation}\label{E:A-TR0}
\psi_t^{\sigma}({ q})=\frac{U_{{ q}}(F_{(z,\sigma)},t)-y}{\delta\sigma z-(1-\delta(1-\sigma))y}.
\end{equation}
By \eqref{E:A-TR1},
\begin{equation}\label{E:A-TR2}
\psi_t^{\sigma}({ q})=(1- q_t)\big(1+\delta(1-\sigma)\psi_{t+1}^{\sigma}({ q})\big)=\sum_{k=0}^\infty \left(\delta^k(1-\sigma)^k\prod_{s=t}^{t+k}(1- q_s)\right).
\end{equation}
Note that for any constant sequence $q^\infty=(q,q,...)$,
\begin{equation}\label{E:A-TR3}
\psi^{\sigma}(q^\infty)=\sum_{k=0}^\infty \delta^k(1-\sigma)^k (1- q)^{k+1}=\frac{1- q}{1-\delta(1-\sigma)(1- q)},
\end{equation}
where we omit the subscript $t$ for notational simplicity. When $\delta\sigma z-(1-\delta(1-\sigma))y>0$, the constant sequence $\bar{q}^\infty$ is better than ${q'}$ if 
\[
\psi^{\sigma}(\bar {q}^\infty)\ge \inf_t \psi_t^{\sigma}({q}').
\]
When $\delta\sigma z-(1-\delta(1-\sigma))y<0$, the constant sequence $\bar{q}^\infty$ is better than ${ q'}$ if 
\[
-\psi^{\sigma}(\bar {q}^\infty)\ge \inf_t (-\psi_t^{\sigma}({ q}')).
\]
Therefore, to prove that $\bar{q}^\infty$ is better than ${q'}$ for all $z\ge y$ and all $\sigma\in[0,1]$, it remains to show that, for all $\sigma \in[0,1]$,
\begin{equation}\label{E:A-TR5}
\inf_t \psi_t^{\sigma}({ q}')\le \psi^{\sigma}(\bar { q}^\infty)\le \sup_t \psi_t^{\sigma}({ q}').
\end{equation}
Fix $\sigma\in[0,1]$. To prove the above inequalities, we first find the interval of values $\psi^0_0({ q})$ achievable by choosing a sequence ${ q}$ subject to the constraint 
\begin{equation}\label{E:Feasibility4}
\inf_t \psi^{\sigma}_t({ q}')\le \psi_s^{\sigma}({ q})\le \sup_t \psi_t^{\sigma}({ q}') \ \ \text{for all $s=0,1,2,...$.}
\end{equation}
To do this, we solve 
\begin{align}
&\min_{ q} \psi^0_0({ q}) \quad \text{subject to \eqref{E:Feasibility4}, and} \label{eredmin}\\
&\max_{ q} \psi^0_0({ q}) \quad \text{subject to \eqref{E:Feasibility4}.} \label{eredmax}
\end{align}

\begin{lemma}\label{L:Seq}
There exist a solution ${ q}^{\sigma}_{\min}$ of \eqref{eredmin} and a solution  ${ q}^{\sigma}_{\max}$ of \eqref{eredmax} that are constant sequences.
\end{lemma}
We postpone the proof of this lemma to the end of this section
and first complete the proof of Proposition \ref{P:StatB}.

By  \eqref{E:A-TR0} and the definition of $\bar q^\infty$ (see \eqref{E:Benchmark}), $\psi^0(\bar { q}^\infty)= \psi^0_0({ q}')$ (recall that we omit the subscript $t$ for constant sequences). Because ${ q}'$ satisfies constraint \eqref{E:Feasibility4} by definition, we have
\begin{equation*}\label{psi0}
\psi^0({ q}^{\sigma}_{\min})\le \psi^0(\bar {q}^\infty)=\psi^0_0({q}')\le \psi^0({ q}^{\sigma}_{\max}).
\end{equation*}
By Lemma \ref{L:Seq}, ${ q}^{\sigma}_{\min}$ and ${ q}^{\sigma}_{\max}$ are constant sequences. By \eqref{E:A-TR3}, for any constant sequence $\tilde q^\infty=(\tilde q,\tilde q,...)$, $\psi^0(\tilde q^\infty)$ is strictly increasing in $\tilde q$. Thus, we have
\begin{equation}\label{E:opop}
{ q}_{\min}^{\sigma}\le \bar { q}^\infty\le { q}_{\max}^{\sigma}.
\end{equation}
Again by \eqref{E:A-TR3},  for any constant sequence $\tilde q^\infty$ and any $\sigma$, $\psi^\sigma(\tilde q^\infty)$ is strictly increasing in $\tilde q$. Since ${ q}_{\min}^{\sigma}$ and ${ q}_{\max}^{\sigma}$ satisfy the constraint \eqref{E:Feasibility4}, we have
\[
\inf_t \psi_t^{\sigma}({ q}')\le \psi^{\sigma}({ q}^{\sigma}_{\min})\le \psi^{\sigma}(\bar { q}^\infty)\le \psi^{\sigma}({ q}^{\sigma}_{\max})\le \sup_t \psi_t^{\sigma}({ q}').
\]
So, \eqref {E:A-TR5} holds. This completes the proof.

\begin{proof}[Proof of Lemma \ref{L:Seq}]
We prove that a solution of the maximization problem \eqref{eredmax} is a constant sequence. The proof of this statement for the minimization problem \eqref{eredmin}  is analogous.

Fix $\sigma\in[0,1]$. We use the notation
\[
\ul\psi^{\sigma}({ q}')=\inf_t \psi_t^{\sigma}({ q}') \quad \text{and} \quad \ol\psi^{\sigma}({ q}')=\sup_t \psi_t^{\sigma}({ q}').
\]
Let $\tilde q$ be the solution of the equation
\begin{equation}\label{E:A-TR6}
\ol\psi^{\sigma}({ q}')=(1-\tilde q)(1+\delta(1-\sigma)\ol\psi^{\sigma}({ q}')).
\end{equation}
We now show that the constant sequence $\tilde { q}^\infty=(\tilde q,\tilde q,...)$ is a solution of the maximization problem \eqref{eredmax}. To prove this, we solve a finite-horizon problem described below. We assume that the individual makes decisions in rounds $t=0,1,...,T$, after which the individual's behavior is fixed by $ q_{t}=\tilde q$ for all $t>T$. Because the maximal value of $\psi^0_0({ q})$ in the problem \eqref{eredmax} can differ from that in the problem with horizon $T$ by at most $\delta^T$, we find the solution to the infinite-horizon problem  \eqref{eredmax} as the limit of the solutions to the finite-horizon problem as $T\to\infty$.

For each $T=1,2,...$ consider the following problem:
\begin{equation}\label{eredmax1}
\begin{split}
&\max_{{ q}} \psi^0_0({ q}) \ \ \text{subject to}\\ 
&\ul\psi^{\sigma}({ q}')\le \psi_t^{\sigma}({ q})\le \ol\psi^{\sigma}({ q}') \ \ \text{for all $t$},\\
& q_t=\tilde q \ \ \text{for all $t=T+1,T+2,...$}.
\end{split}
\end{equation}
We now show that $\tilde {q}^\infty$ is a solution of \eqref{eredmax1}. We proceed by induction, starting from round $k=T$, and then continue to rounds $k=T-1,T-2,...,1,0$. 

Let $k\in\{0,1,...,T\}$ and suppose $ q_t= \tilde q$ for each $t>k$. Observe that, by \eqref{E:A-TR3}, for all $t>k$,
\begin{equation}\label{E:A-TR7}
\psi_t^{0}({ q})=\frac{1-\tilde  q}{1-\delta(1-\tilde  q)} \quad\text{and}\quad
\psi_t^{\sigma}({ q})=\frac{1-\tilde  q}{1-\delta(1-\sigma)(1-\tilde  q)}=\ol\psi^\sigma({ q}'),
\end{equation}
where the last equality is by the definition of $\tilde  q$ in \eqref{E:A-TR6}. 
Next, ${ q}$ must satisfy the constraint in \eqref{eredmax1}, so $\psi_k^{\sigma}({ q})\le \ol\psi^\sigma({ q}')$. Using \eqref{E:A-TR2} and \eqref{E:A-TR7}, we obtain that
\begin{equation}\label{E:A-TR9}
\psi_k^{\sigma}({ q})=(1- q_k)(1+\delta(1-\sigma)\psi_{k+1}^{\sigma}({ q}))=(1- q_k)(1+\delta(1-\sigma)\ol\psi^\sigma({ q}'))\le \ol\psi^\sigma({ q}')
\end{equation}
implies by \eqref{E:A-TR6}
\begin{equation}\label{E:A-TR8}
 q_k\ge \tilde  q.
\end{equation}

Let us first deal with the case of $k\ge 1$. 
We show that if $ q_k>\tilde  q$, then $\psi^0_0({ q})$ can be increased by reducing $ q_k$. Specifically, we keep $ q_t$ fixed for all $t$ different from $k-1$ and $k$, and vary $ q_{k-1}$ and $ q_k$ such that $\psi_{k-1}^{\sigma}({ q})$ remains constant, that is,
\[
d\psi_{k-1}^{\sigma}({ q})=-(1+\delta(1-\sigma)\psi_{k}^{\sigma}({ q}))d q_{k-1}+(1- q_{k-1})\delta(1-\sigma)\frac{\partial \psi_{k}^{\sigma}({ q})}{\partial  q_k}d q_k=0.
\]
By \eqref{E:A-TR7} and \eqref{E:A-TR9} we have
\begin{align*}
\psi_k^{\sigma}({ q})&=(1- q_k)(1+\delta(1-\sigma)\psi_{k+1}^{\sigma}({ q}))=(1- q_k)(1+\delta(1-\sigma)\ol\psi^\sigma({ q}'))\\
&=(1- q_k)\left(1+\delta(1-\sigma)\frac{1-\tilde  q}{1-\delta(1-\sigma)(1-\tilde  q)}\right)=\frac{1- q_k}{1-\delta(1-\sigma)(1-\tilde  q)}
\end{align*}
and
\[
\frac{\partial \psi_k^\sigma({ q})}{\partial  q_k}=-\frac{1}{1-\delta(1-\sigma)(1-\tilde q)}.
\]
Thus,
\[
\frac{d q_{k-1}}{d q_k}=-\frac{\delta(1-\sigma)(1- q_{k-1})}{1-\delta(1-\sigma)( q_{k}-\tilde q)}.
\]
Inserting $\sigma=0$ into \eqref{E:A-TR2} and \eqref{E:A-TR7}, by the induction assumption that $q_{k+1}=\tilde q^\sigma$,
\[
\psi^0_{k}({ q})=(1- q_k)(1+\delta\psi^0_{k+1}({ q}))=(1- q_k)\left(1+\frac{\delta(1-\tilde  q)}{1-\delta(1-\tilde  q)}\right) =\frac{1- q_k}{1-\delta(1-\tilde  q)},
\]
and
\[
\psi^0_{k-1}({ q})=(1- q_{k-1})(1+\delta\psi^0_{k}({ q}))=(1- q_{k-1})\frac{1-\delta( q_k-\tilde  q)}{1-\delta(1-\tilde  q)}.
\]
Thus, by \eqref{E:A-TR2} with $\sigma=0$,
\begin{align*}
\frac{\partial  \psi^0_0({ q})}{\partial  q_{k}}&=\delta^{k}\left(\prod_{s=0}^{k-1}(1- q_s)\right)\frac{\partial  \psi^0_k({ q})}{\partial  q_{k}}=-\delta^{k}\left(\prod_{s=0}^{k-1}(1- q_s)\right)\frac{1}{1-\delta(1-\tilde  q)}
\end{align*}
and
\begin{align*}
\frac{\partial  \psi^0_0({ q})}{\partial  q_{k-1}}&=\delta^{k-1}\left(\prod_{s=0}^{k-2}(1- q_s)\right)\frac{\partial  \psi^0_{k-1}({ q})}{\partial  q_{k-1}}=-\delta^{k-1}\left(\prod_{s=0}^{k-2}(1- q_s)\right)\frac{1-\delta( q_k-\tilde q)}{1-\delta(1-\tilde  q)}\\
&=-\delta^{k}\left(\prod_{s=0}^{k-1}(1- q_s)\right)\frac{1-\delta( q_k-\tilde q)}{(1-\delta(1-\tilde  q))\delta(1- q_{k-1})}.
\end{align*}
Therefore, if $ q_{k-1}<1$, then
\begin{align*}
\frac{d \psi^0_0({ q})}{d  q_{k}}&=\frac{\partial  \psi^0_0({ q})}{\partial  q_{k}}+\frac{\partial  \psi^0_0({ q})}{\partial  q_{k-1}}\frac{d q_{k-1}}{d  q_k}\\
&=-\delta^{k}\left(\prod_{s=0}^{k-1}(1- q_s)\right)\left(\frac{1}{1-\delta(1-\tilde q)}+\frac{1-\delta( q_k-\tilde q)}{(1-\delta(1-\tilde  q))\delta(1- q_{k-1})}\frac{d q_{k-1}}{d  q_k}\right)\\
&=-\frac{\delta^{k}}{1-\delta(1-\tilde  q)}\left(\prod_{s=0}^{k-1}(1- q_s)\right)\left(1-\frac{(1-\delta( q_k-\tilde q))(1-\sigma)}{1-\delta(1-\sigma)( q_{k}-\tilde q)}\right)\\
&=-\frac{\delta^{k}}{1-\delta(1-\tilde  q)}\left(\prod_{s=0}^{k-1}(1- q_s)\right)\frac{\sigma}{1-\delta(1-\sigma)( q_{k}-\tilde q)}\le 0.
\end{align*}
Alternatively, if $ q_{k-1}=1$, then $\psi^0_{0}({ q})$ is independent of $ q_k$, so $d \psi^0_0({ q})/d  q_{k}=0$. Thus, if $ q_k>\tilde  q$, then decreasing $ q_k$ increases $\psi_0^0({ q})$ without violating the constraint in \eqref{eredmax1}, as long as $ q_k\ge\tilde q$. 

Next, we deal with the case of $k=0$. By \eqref{E:A-TR2} and \eqref{E:A-TR7} we have
\[
\frac{d \psi_0^0({ q})}{d  q_0}=-1-\delta\psi_1^0({ q})<0.
\]
So, again, if $ q_0>\tilde  q$, then decreasing $ q_0$ increases $\psi_0^0({ q})$ without violating the constraint in \eqref{eredmax1}, as long as $ q_0\ge\tilde q$. 

We thus proved that if ${ q}$ is a solution of \eqref{eredmax1} with $ q_k>\tilde q$ and $ q_t=\tilde q$ for all $t>k$, then there exists a solution with  $ q_t=\tilde q$ for all $t\ge k$. As this is true for each $k=T,T-1,...,1,0$ by induction, we obtain that $\tilde { q}^\infty$ is a solution of \eqref{eredmax1}, so ${ q}_{\max}^{\sigma}=\tilde { q}^\infty$.
\end{proof}

\section*{Appendix B. General Environments}
\renewcommand{\thesection}{B}
\setcounter{subsection}{0}
\subsection{Proof of Proposition \ref{P:Stat}}\label{s:pstat}
We begin the proof with a lemma that will be useful here and in further proofs in Appendix B. 
\begin{lemma} Let $p$ be stationary. For each $y\ge x_0$ and each $F_{(z,\sigma)}\in\mathcal B_X$,
\begin{align}
&U_p(F_{(z,\sigma)},y)=\frac{p(y) y}{1-\delta(1-p(y))} &\text{if $z\le y$,}\label{E:U-Z}\\
&U_p(F_{(z,\sigma)},y)=\frac{p(y)y+(1-p(y))\delta \sigma \frac{p(z) z}{1-\delta(1-p(z))}}{1-\delta(1-p(y))(1-\sigma)} &\text{if $z>y$,}\label{E:U-Y}\\
&\frac{U_p(F_{(z,\sigma)},y)}{V(F_{(z,\sigma)},y)}=\frac{p(y)y+(1-p(y))\delta\sigma \frac{p(z)z}{1-\delta(1-p(z))}}{(1-\delta(1-p(y))(1-\sigma))\frac{\delta\sigma z}{1-\delta(1-\sigma)}} &\text{if $c_{F_{(z,\sigma)}}\ge y$.}\label{E:PR-1}
\end{align}
Moreover, if $p(y)$ is monotone, then
\begin{align}
&\frac{U_p(F_{(z,\sigma)},y)}{V(F_{(z,\sigma)},y)}\ge\frac{U_p(F_{(z,0)},y)}{V(F_{(z,0)},y)}=\frac{p(y)}{1-\delta(1-p(y))} &\text{if $c_{F_{(z,\sigma)}}\le y$.} \label{E:PR-0}
\end{align}
\end{lemma}
\begin{proof}
If $z\le y$, then the best-so-far alternative never changes under $F_{(z,\sigma)}$. The payoff is $U_p(F_{(z,\sigma)},y)=p(y)y+(1-p(y))\delta U_p(F_{(z,\sigma)},y)$. Solving this equation for $U_p(F_{(z,\sigma)},y)$ yields \eqref{E:U-Z}. 
Alternatively, if $z>y$, then the payoff is
\begin{equation}\label{FundiIneq}
U_p(F_{(z,\sigma)},y)=p(y)y+(1-p(y))\delta(\sigma U_p(F_{(z,\sigma)},z)+(1-\sigma) U_p(F_{(z,\sigma)},y)).
\end{equation}
Inserting $y=z$ into \eqref{E:U-Z} yields $U_p(F_{(z,\sigma)},z)=\frac{p(z) z}{1-\delta(1-p(z))}$. Inserting this into \eqref{FundiIneq} and solving for $U_p(F_{(z,\sigma)},y)$ yields \eqref{E:U-Y}.
To prove \eqref{E:PR-1}, suppose that $c_{F_{(z,\sigma)}}=\frac{\delta\sigma z}{1-\delta(1-\sigma)}\ge y$. Note that $z\ge y/\delta>y$, since $\sigma\in[0,1]$. So, $U_p(F_{(z,\sigma)},y)$ is given by \eqref{E:U-Y} and $V(F_{(z,\sigma)},y)=\frac{\delta\sigma z}{1-\delta(1-\sigma)}$ by \eqref{E:RVS}, and  \eqref{E:PR-1} follows immediately.

Finally, to prove \eqref{E:PR-0}, suppose that $c_{F_{(z,\sigma)}}\le y$. Observe that $V(F_{(z,\sigma)},y)=y$ by \eqref{E:RVS}. If $z\le y$, then by \eqref{E:U-Z} the payoff ratio is
\[
\frac{U_p(F_{(z,\sigma)},y)}{V(F_{(z,\sigma)},y)}=\frac{p(y)}{1-\delta(1-p(y))}.
\]
Instead, if $z>y$, then, using \eqref{E:U-Y} and $p(z)\ge p(y)$ by the monotonicity of $p$, we have
\begin{align*}
\frac{U_p(F_{(z,\sigma)},y)}{V(F_{(z,\sigma)},y)}&=\frac{p(y)y+(1-p(y))\delta \sigma \frac{p(z) z}{1-\delta(1-p(z))}}{(1-\delta(1-p(y))(1-\sigma))y}\\ &\ge\frac{p(y)y+(1-p(y))\delta \sigma \frac{p(y) y}{1-\delta(1-p(y))}}{(1-\delta(1-p(y))(1-\sigma))y}=\frac{p(y)}{1-\delta(1-p(y))}.
\end{align*}
\end{proof}

We now prove Proposition \ref{P:Stat}. Fix $X\subset \R$. Let $\hat R_p(x_0)$ denote the smallest payoff ratio when facing an environment in $\mathcal B_X$, so
\[
\hat R_p(x_0)=\inf_{h\in \mathcal H(x_0)}\inf_{F\in {\mathcal B}_X} \frac{U_p(F,h)}{V(F,h)}.
\]
Note that $\hat R_p(x_0)$ is not that same value as the performance ratio $R_p(x_0,\mathcal B_X)$ calculated for rule $p$ in the binary setting. This is because in the binary setting the individual's choice is trivial whenever she observes an alternative above the outside option. 

To prove Proposition \ref{P:Stat}, we show that for each rule $p$ there exists a rule $q$ that is stationary, monotone, and has the monotone ratio property, such that
\[
R_p(x_0,\mathcal F_X)\le \hat R_p(x_0)\le \hat R_{q}(x_0)=R_{q}(x_0, \mathcal F_X).
\]
The first inequality trivially follows from ${\mathcal B}_X\subset \mathcal F_X$ and the definitions of $R_p$ and $\hat R_p$. Proposition \ref{P:Binary} proves the equality, $\hat R_{q}(x_0)=R_{q}(x_0, \mathcal F_X)$. We hasten to point out that the proof of Proposition \ref{P:Binary} does not depend on Proposition \ref{P:Stat}. It remains to prove that
\[
\hat R_p(x_0)\le \hat R_{q}(x_0).
\]

We divide the proof into three parts. In each part, we consider a decision rule $p$ that satisfies the restrictions imposed in the previous parts, and construct a different rule whose performance ratio over the set of binary environments is weakly better than that of $p$. 

\noindent{\it Part 1. Stationarity.} Let $p$ be a decision rule. We now construct a stationary rule $q$ whose performance ratio against environments in $\mathcal B_X$ is at least as high as that of $p$.

Let $\bar {\mathcal H}(y)$ denote the set of histories whose best-so-far alternative is $y$. Let $W(y)$ be the maximal payoff of rule $p$ among all these histories, against all binary environments in which no alternatives better than $y$ will ever emerge, so
\begin{equation}\label{E:W-Z}
W(y)=\sup_{h\in \bar {\mathcal H}(y)}\left(\sup_{F_{(z,\sigma)}\in {\mathcal B}_X: z\le y} U_p(F_{(z,\sigma)},h)\right).
\end{equation}
Define a stationary rule $q$ as follows. For each $y\ge x_0$, let $q(y)$ be the solution of
\[
W(y)=q(y)y+(1-q(y))\delta W(y).
\]
Note there exists a unique solution $q(y)\in[0,1]$, because $y\ge x_0>0$ and, by \eqref{E:RVS},
\begin{equation}\label{E:y0}
0\le  U_p(F_{(z,\sigma)},h)\le  V(F_{(z,\sigma)},h)=y \ \ \text{if $z\le y$ and $h\in\bar{\mathcal H}(y)$,}
\end{equation}
so, in particular, $0\le W(y)\le y$. 
We now prove that the change from $p$ to $q$ does not decrease the performance ratio, so $\hat R_p(x_0)\le \hat R_q(x_0)$.

Fix $y'\ge x_0$ and $F_{(z,\sigma)}\in\mathcal B_X$. 
Denote by $q|_{y'}p$ a decision rule in which the stopping probability is $q(y)$ whenever the best-so-far alternative is $y\ne y'$, and it is given by the original rule  $p(h)$ whenever the best-so-far alternative is $y'$, that is, $h\in\bar{\mathcal H}(y')$. 
We now prove that
\begin{equation}\label{E:yy-1}
\inf_{h'\in \bar{\mathcal H}(y')} U_p(F_{(z,\sigma)},h')\le \inf_{h\in\bar{\mathcal H}(y')} U_{q|_{y'}p}(F_{(z,\sigma)},h)\le U_{q}(F_{(z,\sigma)},y'). 
\end{equation}

To prove the first inequality in \eqref{E:yy-1}, we fix an arbitrary $h'\in\bar{\mathcal H}(y')$ and show that $U_{q|_{y'}p}(F_{(z,\sigma)},h')\ge U_p(F_{(z,\sigma)},h')$. We have
\[
 U_{q|_{y'}p}(F_{(z,\sigma)},h')=p(h')y'+(1-p(h'))\delta ((1-\sigma)U_{q|_{y'}p}(F_{(z,\sigma)},h'\oplus 0)+\sigma  U_{q|_{y'}p}(F_{(z,\sigma)},h'\oplus z)),
\]
where `$\oplus$' denotes the vector concatenation operator, so $h'\oplus 0$ is the vector $h'$ with $0$ appended at the end. 
If $z\le y'$, then $U_{q|_{y'}p}(F_{(z,\sigma)},h')$ is independent of $q$ (because the best-so-far alternative remains $y'$), so
\[
U_{q|_{y'}p}(F_{(z,\sigma)},h')=U_{p}(F_{(z,\sigma)},h').
\]
Otherwise, if $z>y'$, then, by the definitions of $W(z)$ and $q$, for each $k=0,1,...$,
\begin{equation}\label{E:y3}
U_{q|_{y'}p}(F_{(z,\sigma)},h'\oplus {\bf 0}^k\oplus z)=U_q(F_{(z,\sigma)},z)=W(z)\ge U_p(F_{(z,\sigma)},h'\oplus {\bf 0}^k\oplus z),
\end{equation}
where ${\bf 0}^k$ is the vector of $k$ zeros. So,
\begin{eqnarray*}
U_{q|_{y'}p}(F_{(z,\sigma)},h')&=&\sum_{k=0}^\infty \bigg[\big(p(h'\oplus {\bf 0}^k)y'+(1-p(h'\oplus {\bf 0}^k))\delta\sigma W(z)\big)\\
&&\times \delta^k(1-\sigma)^k\prod_{s=0}^{k-1}(1-p(h'\oplus {\bf 0}^s))\bigg]\\
&\ge &\sum_{k=0}^\infty \bigg[\big(p(h'\oplus {\bf 0}^k)y'+(1-p(h'\oplus {\bf 0}^k))\delta\sigma U_p(F_{(z,\sigma)},h'\oplus {\bf 0}^k\oplus z)\big)\\
&&\times \delta^k(1-\sigma)^k\prod_{s=0}^{k-1}(1-p(h'\oplus {\bf 0}^s))\bigg]=U_{p}(F_{(z,\sigma)},h').
\end{eqnarray*}
Summing up the above, we obtain $U_{q|_{y'}p}(F_{(z,\sigma)},h')\ge U_p(F_{(z,\sigma)},h')$ for each $h'\in\bar{\mathcal H}(y')$, thus proving the first inequality in \eqref{E:yy-1}.

Let us prove the second inequality in \eqref{E:yy-1}. If $z\le y'$, then $U_q(F_{(z,\sigma)},y')=W(y')\ge \inf_{h\in\bar{\mathcal H}(y')} U_{q|_{y'}p}(F_{(z,\sigma)},h)$ by the definitions of $W(y')$ and $q(y')$.

Alternatively, let $z>y'$. So, for each $k=0,1,2,...$, as long as $z$ has not been realized, the only possible history is $h'\oplus {\bf 0}^k$. Define
\[
 q_k'=p(h' \oplus {\bf 0}^k), \ \ \text{$k=0,1,2,...$}.
\]
This is the problem with binary environments analyzed in Section \ref{s:binary}, where $x_0=y'$ and $X=\{0,W(z)\}$, so the value of the high alternative is $W(z)$.  By Proposition \ref{P:StatB}, we can replace the sequence of probabilities $(q_0',q_1',...)$ by a constant sequence $\bar q^\infty=(\bar q,\bar q,...)$. Moreover, $\bar q=q(y')$ by \eqref{E:Benchmark} and the definitions of $W(y')$ and $q(y')$. We thus proved the second inequality in \eqref{E:yy-1}.

By \eqref{E:RVS}, $V(F_{(z,\sigma)},h')$ depends on $h'$ only through the best-so-far alternative $y'$, so $V(F_{(z,\sigma)},h')=V(F_{(z,\sigma)},y')$. It follows from \eqref{E:yy-1} that
\[
\inf_{h'\in \bar{\mathcal H}(y')} \frac{U_{p}(F_{(z,\sigma)},h')}{V(F_{(z,\sigma)},h')}=\frac{ \inf\limits_{h'\in \bar{\mathcal H}(y')} U_p(F_{(z,\sigma)},h')}{V(F_{(z,\sigma)},y')}\le \frac{ \inf\limits_{h'\in \bar{\mathcal H}(y')} U_{q|_{y'}p}(F_{(z,\sigma)},h')}{V(F_{(z,\sigma)},y')}\le \frac{U_{q}(F_{(z,\sigma)},y')}{V(F_{(z,\sigma)},y')}.
\]
The above holds for each $y'\ge x_0$ and each $F_{(z,\sigma)}\in\mathcal B_X$, thus proving $\hat R_p(x_0)\le \hat R_q(x_0)$.

\medskip
\noindent{\it Part 2. Monotonicity.} Consider a stationary rule $p$. 
Suppose that $p$ is nonmonotone, so there exist $y',y''$ such that $x_0\le y'<y''$ and $p(y')>p(y'')$. Define $q$ by
\begin{equation}\label{E:px}
q(y)=\sup_{y'\in[x_0,y]} p(y'), \ \ y\ge x_0.
\end{equation}
Note that, for all $y\ge x_0$,
\begin{equation}\label{E:px-1}
q(y)\ge p(y)\quad\text{and}\quad\frac{q(y)}{1-\delta(1-q(y))}\ge \frac{p(y)}{1-\delta(1-p(y))}.
\end{equation}
We now show that this change from $p$ to $q$ does not decrease the performance ratio. 
Consider any $y\ge x_0$. For each $F_{(z,\sigma)}$ such that $c_{F_{(z,\sigma)}}\le y$, 
\[
\frac{U_{q}(F_{(z,\sigma)},y)}{V(F_{(z,\sigma)},y)}\ge \frac{q(y)}{1-\delta(1- q(y))}\ge \frac{p(y)}{1-\delta(1- p(y))}=\frac{U_{p}(F_{(z,0)},y)}{V(F_{(z,0)},y)}\ge r_p(y),
\]
The first inequality is by \eqref{E:PR-0}, where we use the monotonicity of $q$ (by construction). The second inequality is by \eqref{E:px-1}. The equality is by \eqref{E:U-Z} and $V(F_{(z,0)},y)=y$. The last inequality is by the definition of $r_p(y)$ in \eqref{E:r-p}.

Next, for each $F_{(z,\sigma)}$ such that $c_{F_{(z,\sigma)}}>y$, by \eqref{E:PR-1}, 
\[
\frac{U_{q}(F_{(z,\sigma)},y)}{V(F_{(z,\sigma)},y)}=\frac{q(y)y+(1-q(y))\delta \sigma \frac{q(z) z}{1-\delta(1-q(z))}}{\left(1-\delta(1- q(y))(1-\sigma)\right)\frac{\delta\sigma z}{1-\delta(1-\sigma)}}.
\]
By the definition of $q(y)$, there exists $y'\in[x_0,y]$ such that $q(y)=p(y')$. Therefore,
\[
\frac{U_{q}(F_{(z,\sigma)},y)}{V(F_{(z,\sigma)},y)}\ge \frac{p(y')y'+(1-p(y'))\delta \sigma \frac{p(z) z}{1-\delta(1-p(z))}}{\left(1-\delta(1- p(y'))(1-\sigma)\right)\frac{\delta\sigma z}{1-\delta(1-\sigma)}}=\frac{U_{p}(F_{(z,\sigma)},y')}{V(F_{(z,\sigma)},y')}\ge r_p(y').
\]
We thus obtain that, for each $y\ge x_0$, $r_{q}(y)\ge r_{ p}(y')$ for some $y'\in[x_0,y]$. It follows that $\hat R_{q}(x_0)=\inf_{y\ge x_0} r_{q}(y)\ge \hat R_p(x_0)=\inf_{y\ge x_0} r_{p}(y)$. We conclude that, without loss of generality, we can restrict attention to monotone rules.


\medskip
\noindent{\it Part 3. Monotone Ratio Property.} 
Consider a monotone stationary rule $p$. Suppose that $r_{p}(y)$ defined by \eqref{E:r-p} is nonmonotone. 

First, we show that if $\bar x=\sup X<\infty$, then, without loss of generality, we can assume 
\begin{equation}\label{E:rpp}
\begin{split}
&r_{p}(y)=1 \quad \quad \quad \, \text{for all $y\ge \delta \bar x$, and}\\
&r_{p}(y)\ge \frac{y}{\delta \bar x}\ \ \text{for all $y\in[x_0,\delta\bar x]$.}
\end{split}
\end{equation}
The first line is trivial, as when $y\ge \delta\bar x$, one can trivially get the ratio of $1$ by stopping and getting the best-so-far alternative $y$. To show the second line, suppose that $r_{p}(y')<\frac{y'}{\delta \bar x}$ for some $y'$. Then define $q(y)=1$ for all $y\ge y'$ and $q(y)=p(y)$ for all $y<y'$. For each $y\ge y'$,
\[
r_{q}(y)=\frac{y}{\delta \bar x}\ge \frac{y'}{\delta \bar x}>r_{p}(y').
\]
For each $y\in[x_0,y')$, using \eqref{E:PR-1}, the definition of $r_{p}(y)$, and $\frac{q(z)}{1-\delta(1-q(z))}\ge \frac{p(z)}{1-\delta(1-p(z))}$, we obtain $r_{q}(y)\ge r_{p}(y)$. It follows that $\hat R_{q}(x_0)=\inf_{y\ge x_0} r_{q}(y)\ge \hat R_p(x_0)=\inf_{y\ge x_0} r_{p}(y)$. 


As $r_{p}(y)$ is nonmonotone, there exists $y'$ and $y''$ such that $\delta y''\le y'<y''$ and $r_{p}(y')>r_p(y'')=\inf_{y\ge y' }r_p(y)$. We now construct a monotone stationary rule $q(y)$ that differs from $p(y)$ only on the interval $[y',y'')$ and has the following properties: $r_{q}(y)$ is constant on $[y',y'')$, continuous at $y''$, and satisfies $\hat R_q(x_0)\ge \hat R_p(x_0)$. Let
\[
D(y,g)=\min\left\{\frac{g}{1-\delta(1-g)},\inf_{\substack{z>y'',\\ \sigma\in[0,1]}}\frac{g y+(1-g) \delta\sigma\frac{p(z) z}{1-\delta(1-p(z))}}{(1-\delta(1-g)(1-\sigma))\frac{\delta\sigma z}{1-\delta(1-\sigma)}}\right\}.
\]
Note that  $D(y,p(y))=r_{p}(y)$ for each $y\in[y',y'')$. This is by \eqref{E:PR-1} and \eqref{E:PR-0}, and the fact that $c_{F_{(z,\sigma)}}>y$ implies $z> y/\delta\ge y''$. Since it is assumed that $r_{p}(y)>r_p(y'')$ for each $y\in[y',y'')$, we have 
\[
D(y,p(y))>D(y'',p(y'')) \ \ \text{for each $y\in[y',y'')$.}
\]
Next, we have
\[
\frac{d}{dg} \left(\frac{g y+(1-g) \delta\sigma\frac{p(z) z}{1-\delta(1-p(z))}}{(1-\delta(1-g)(1-\sigma))\frac{\delta\sigma z}{1-\delta(1-\sigma)}}\right)=\frac{y(1-\delta(1-\sigma))-\delta\sigma \frac{p(z)z}{1-\delta(1-p(z))}}{\left((1-\delta(1-g)(1-\sigma))\frac{\delta\sigma z}{1-\delta(1-\sigma)}\right)^2},
\]
which has a sign that does not depend on $g$. So $\frac{g y+(1-g) \delta\sigma\frac{p(z) z}{1-\delta(1-p(z))}}{(1-\delta(1-g)(1-\sigma))\frac{\delta\sigma z}{1-\delta(1-\sigma)}}$ is monotone in $g$ for each $z$, $\sigma$ and $y$. Thus, $D(y,g)$ is a lower envelope of monotone functions, so it is quasiconcave in $g$ for each $y$. Moreover, 
\[
D(y,1)=\frac{y}{\delta \bar x}<D(y'',p(y''))=r_p(y'') \ \ \text{for each $y\in[y',y'')$,}
\]
because, by \eqref{E:rpp}, $r_p(y'')\ge \frac{y}{\delta \bar x}$. To sum up,
\[
D(y,1)<D(y'',p(y''))<D(y,p(y))  \ \ \text{for each $y\in[y',y'')$.}
\]
Since $D(y,g)$ is continuous and quasiconcave in $g$, for each $y\in[y',y'')$ there exists $g^*(y)\ge p(y)$ such that $D(y,g^*(y))=D(y'',p(y''))$.
Moreover, since $D(y,g)$ is increasing in $y$ for all $g$, by the monotone comparative statics theorem (\citeauthor{Milgrom} 1994, Theorem 4$'$), $g^*(y)$ is increasing.

Define a stationary rule $q$ as follows. For each $y\in[y',y'')$, let $q(y)=g^*(y)$, and for each $y\not\in[y',y'')$, let $q(y)=p(y)$. We thus obtain
\begin{equation}\label{E:Ineqpq}
q(y)\ge p(y) \quad \text{and} \quad \frac{q(y)}{1-\delta(1-q(y))}\ge \frac{p(y)}{1-\delta(1-p(y))} \ \ \text{for each $y\ge x_0$},
\end{equation}
and
\begin{equation}\label{E:Ineqpq1}
r_q(y)=r_p(y'')  \ \ \text{for each $y\in[y',y'']$.}
\end{equation}
Therefore, $r_q(y)$ is monotone on $[y',y'']$. Moreover, for each $y<y'$, by \eqref{E:PR-1} and \eqref{E:Ineqpq}, $r_{q}(y)\ge r_{p}(y)$. For each $y\in[y',y'']$, by \eqref{E:Ineqpq1}, $r_{q}(y)=r_{p}(y'')$. For each $y>y''$, by $q(y)=p(y)$, $r_{q}(y)=r_{p}(y)$.  We thus obtain that, for each $y\ge x_0$, $r_{q}(y)\ge \min\{r_{ p}(y),r_{ p}(y'')\}$. It follows that $\hat R_{q}(x_0)=\inf_{y\ge x_0} r_{q}(y)\ge \hat R_p(x_0)=\inf_{y\ge x_0} r_{p}(y)$. 
 
We thus conclude that, without loss of generality, we can restrict attention to rules $p$ such that $r_{p}(y)$ is weakly increasing in $y$.
This completes the proof.

\subsection{Proof of Proposition \ref{P:Binary}}\label{s:pbin}
Let $p$ be stationary, monotone, and satisfy the monotone ratio property. We now prove that
\[
R_{p}(x_0,\mathcal F_X)\ge\hat R_{p}(x_0)=\inf_{y\ge x_0}\inf_{F\in {\mathcal B}_X} \frac{U_p(F,y)}{V(F,y)}.
\] 
Let $\bar x=\sup X$. Note that $\bar x=\infty$ if $X$ is unbounded. Fix a best-so-far alternative $y$ such that $y\ge x_0$. Consider an arbitrary environment $G\in\mathcal F_X$, and denote its reservation value by $c$, so, by \eqref{E:RV},
\begin{equation}\label{E:RV2}
\int_{0}^{c} c \df G(x)+\int_{c}^{\bar x} x\df G(x)=\frac{c}{\delta}.
\end{equation}
We now find a binary environment $F_{(z,\sigma)}$ such that
\[
\frac{U_p(F_{(z,\sigma)},y)}{V(F_{(z,\sigma)},y)}\le \frac{U_p(G,y)}{V(G,y)}.
\]

Let us denote by $F_{(w,z,\sigma)}$ the lottery between $w$ and $z$ with probabilities $1-\sigma$ and $\sigma$, respectively. 
The construction of  $F_{(z,\sigma)}$ consists of two parts. In Part 1, we find $F_{(w,z,\sigma)}$ such that $U_p(F_{(w,z,\sigma)},y)\le U_p(G,y)$ and $V(F_{(w,z,\sigma)},y)\ge V(G,y)$. In Part 2, we show that $w=0$ is without loss of generality.

\medskip\noindent {\it Part 1.} Consider a stationary rule $p$. Fix a best-so-far alternative $y$ such that $y\ge x_0$. Consider an arbitrary environment $G\in\mathcal F_X$, and denote its reservation value by $c$, so, by \eqref{E:RV},
\begin{equation}\label{E:RV2}
\int_{0}^{c} c \df G(x)+\int_{c}^{\bar x} x\df G(x)=\frac{c}{\delta}.
\end{equation}
We now find an environment $F_{(w,z,\sigma)}$ such that
\[
\frac{U_p(F_{(w,z,\sigma)},y)}{V(F_{(w,z,\sigma)},y)}\le \frac{U_p(G,y)}{V(G,y)}.
\]
To find such $F_{(w,z,\sigma)}$, we first consider a one-shot deviation to some environment $F$ under the constraint $V(F,y)\ge V(G,y)$. The ``one-shot deviation'' means that the individual will face $F$ in the next round, and $G$ in all subsequent rounds. We will show that there exists $F=F_{(w,z,\sigma)}$ such that the individual's expected payoff against the sequence of environments $(F_{(w,z,\sigma)},G,G,...)$ is weakly lower than against the original i.i.d.~sequence $(G,G,G,...)$. We then show that this expected payoff is even lower if we replace $(F_{(w,z,\sigma)},G,G,...)$ by $(F_{(w,z,\sigma)},F_{(w,z,\sigma)},F_{(w,z,\sigma)},...)$, thus proving $U_p(F_{(w,z,\sigma)},y)\le U_p(G,y)$.

 Recall that
\begin{align}
U_p(G,y) &=p(y) y+(1-p(y))\delta \int_0^{\bar x} U_p(G,\max\{y,x\}) \df G(x).\label{L3-0}
\end{align}
Let us find a distribution $F$ that minimizes the individual's expected payoff against all one-shot deviation sequences $(F,G,G,...)$, subject to the constraint that the reservation value of $F$ is at least $c$ (which implies $V(F,y)\ge V(G,y)$):
\begin{align}
&\inf_{F\in \mathcal F_{X}} p(y) y+(1-p(y))\delta \int_0^{\bar x} U_p(G,\max\{y,x\}) \df F(x) \label{E:Gamma}\\
&\text{s.t.} \ \int_{0}^{c} c \df F(x)+\int_{c}^{\bar x} x\df F(x)\ge\frac{c}{\delta}.\label{E:GammaC}
\end{align}
Observe that the constraint \eqref{E:GammaC} does not impose any restrictions on how a mass $F(c)$ is assigned to the interval $[0,c]$. Thus, any positive mass $F(c)$ should be assigned to a point $x'$ that minimizes $U_p(F,\max\{y,x'\})$ on $[0,c]\cap X$. So, we can simplify the problem \eqref{E:Gamma}--\eqref{E:GammaC} by using the notation
\[
u(x)=\begin{cases}
\inf_{x'\in [0,c]\cap X} U_p(G,\max\{y,x'\}), & x=c,\\ 
U_p(G,\max\{y,x\}), &x>c, \, x\in X,
\end{cases}
\]
where $u(x)$ is linearly extended to $(c,\bar x)\backslash X$, so $u(x)$ is defined on $[c,\bar x]$.
The problem \eqref{E:Gamma}--\eqref{E:GammaC}  reduces to
\begin{align}
&\inf_{F\in \mathcal F_{X}} \int_{c}^{{\bar x}}u(x)\df F(x) 
\ \ \text{s.t.} \ \int_{c}^{\bar x} x\df F(x)\ge\frac{c}{\delta}.\label{E:Gamma1}
\end{align}

\begin{figure}
\begin{tabular}{ccc}
\hspace*{-0.3in}\includegraphics[width=150pt]{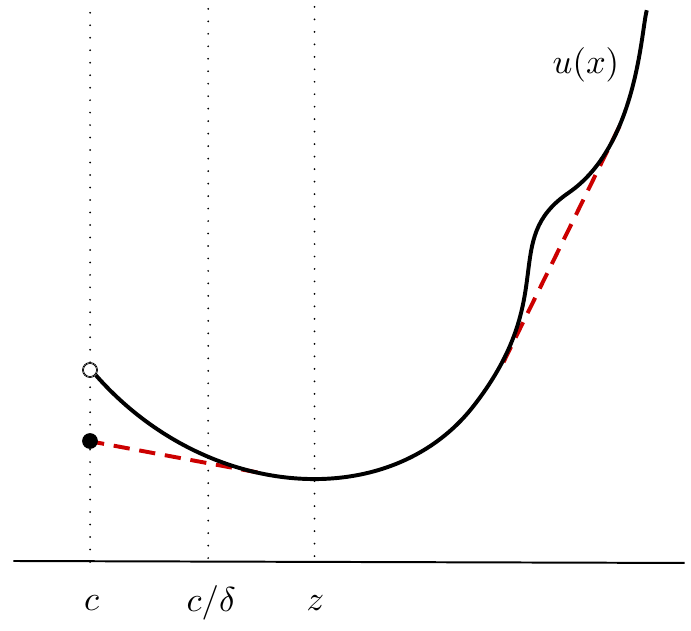} & \includegraphics[width=150pt]{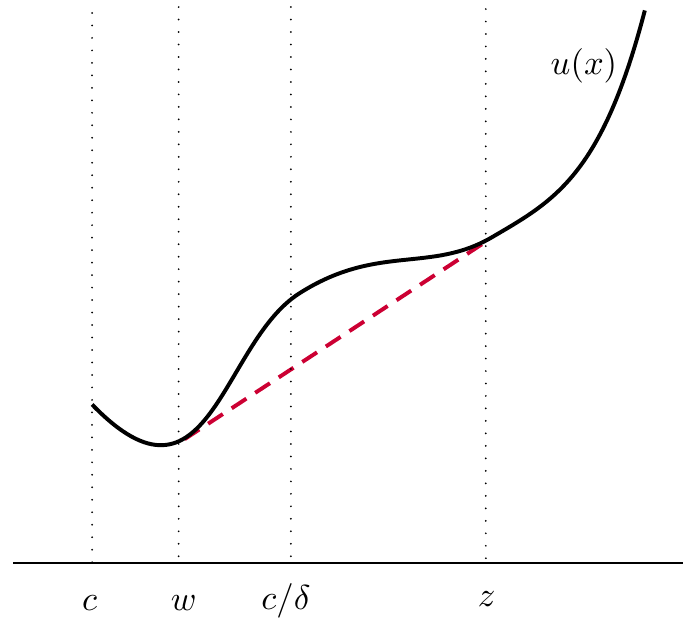} & \includegraphics[width=150pt]{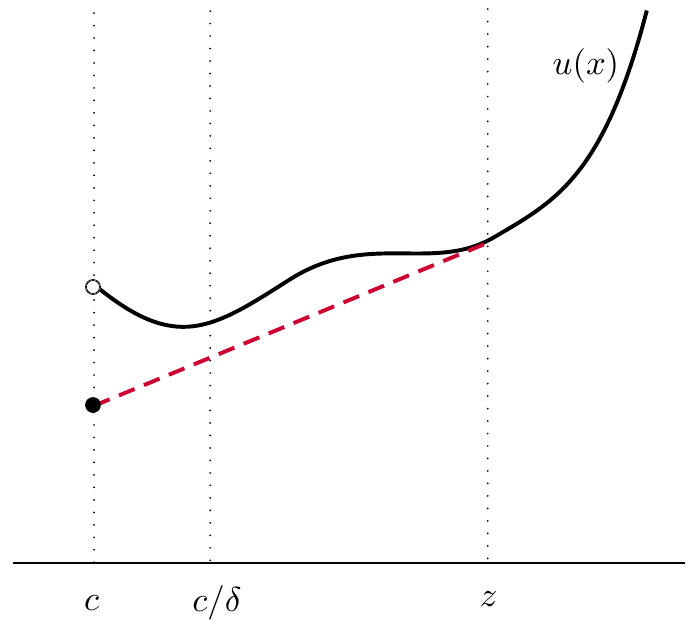} \\
\footnotesize{(a)}&\footnotesize{(b)} &\footnotesize{(c)}
\end{tabular}
\caption{\small Illustration of three cases that arise when solving \eqref{E:Gamma1}.}\label{f:1}
\end{figure}

This problem is solved by the convexification method as in \citeasnoun{KG}, by minimizing the convex closure of $u$ (i.e., the supremum among all continuous and convex functions that do not exceed $u$) on the set $[c,\bar x]$, and thus yielding a solution with a support on at most two points, $w$ and $z$. Figure~\ref{f:1} illustrates how such a solution is found for three different shapes of $u(x)$. The solid curve is $u(x)$, which can be discontinuous at $c$, and the dashed line is the convex closure of $u(x)$ where it is different from $u(x)$. In Figure~\ref{f:1}(a) the minimum of $u(x)$ is attained at $z\ge c/\delta$, so the solution puts the unit mass on the single point $z$. In Figs.~\ref{f:1}(b) and \ref{f:1}(c) the minimum of $u(x)$ is below $c/\delta$, so the solution minimizes the convex closure of $u(x)$ at $x=c/\delta$. In Figure~\ref{f:1}(b) it is obtained by a convex combination of two points, $w$ and $z$, and in Figure~\ref{f:1}(c) it is obtained by a convex combination of $c$ and $z$ as shown on the picture. Note that in the last case, to solve the problem \eqref{E:Gamma}--\eqref{E:GammaC}, one must replace $c$ with a point $w\le c$ where the value $u(c)=\inf_{x'\in [0,c]} U_p(G,\max\{y,x'\})$ is achieved.

Let us formalize the above. For every $\eps>0$ there exists $(w,z,\sigma)$ such that 
\begin{align}
\int_c^\infty u(x)\df F(x)&\ge \int_c^\infty u(x)\df F_{(w,z,\sigma)}(x)-\eps\notag\\
&=(1-\sigma)u(\max\{w,c\})+\sigma u(z)-\eps,\label{L3-1}
\end{align}
and $F_{(w,z,\sigma)}$ satisfies the constraint in \eqref{E:Gamma1}, so
\begin{equation}\label{L3-10}
(1-\sigma)\max\{w,c\}+\sigma z\ge c/\delta.
\end{equation}
Therefore, $U_p(G,y)+\eps$ is weakly greater than the individual's expected payoff against the sequence of environments  $(F_{(w,z,\sigma)},G,G,...)$, where $F_{(w,z,\sigma)}$ satisfies \eqref{L3-1} and \eqref{L3-10}.

We now show that the individual's expected payoff is even lower if we replace $(F_{(w,z,\sigma)}$, $G,G,...)$ by $(F_{(w,z,\sigma)},F_{(w,z,\sigma)},F_{(w,z,\sigma)},...)$.

\medskip
\noindent {\it Case 1.} Suppose that $u(x)$ attains its infimum at or above $c/\delta$, that is,  $\inf_{x\in[c,c/\delta)} u(x)\ge \inf_{x\in[c/\delta,\bar x]}  u(x)$, as shown in Figure~\ref{f:1}(a). Then the constraint in \eqref{E:Gamma1} is not binding. So, for every $\eps>0$ we can find $z_\eps\ge c/\delta$ such that $F_{(0,z_\eps,1)}$ that assigns the unit mass on $z_\eps$ satisfies \eqref{L3-1}. 

By \eqref{L3-0}, \eqref{L3-1}, and the definition of $u(x)$ we have
\begin{align*}
u(z_\eps)\ge p(z_\eps)z_\eps+(1-p(z_\eps))\delta (u(z_\eps)-\eps).
\end{align*}
Solving the inequality for $u(z_\eps)-\eps$, we have $u(z_\eps)-\eps\ge \frac{p(z_\eps) z_\eps-\eps}{1-\delta(1-p(z_\eps))}.$
Therefore,
\begin{align*}
U_p(G,y)&\ge u(y)\ge p(y)y+(1-p(y))\delta(u(z_\eps)-\eps)\ge p(y)y+(1-p(y))\delta \frac{p(z_\eps) z_\eps-\eps}{1-\delta(1-p(z_\eps))}\\
&=U_p(F_{(0,z_\eps,1)},y)-(1-p(y))\delta\frac{ \eps}{1-\delta(1-p(z))}\ge U_p(F_{(0,z_\eps,1)},y)-\frac{\delta \eps}{1-\delta}.
\end{align*}
Also, observe that, by \eqref{OptV} and $z_\eps\ge c/\delta$,
\[
V(G,y)=\max\{y,c\}\le \max\{y,\delta z_\eps\}=V(F_{(0,z_\eps,1)},y).
\]
Thus, for every $\eps>0$ we find $z_\eps$ such that
\begin{equation}\label{E:Conc1}
\frac{U_p(G,y)}{V(G,y)}\ge \frac{U_p(F_{(0,z_\eps,1)},y)-\frac{\delta \eps}{1-\delta}}{V(F_{(0,z_\eps,1)},y)}.
\end{equation}
In particular, $F_{(0,z_\eps,1)}$ satisfies  \eqref{L3-1} and \eqref{L3-10} if one replaces $\eps$ by $\delta\eps/(1-\delta)$.

\medskip
\noindent{\it Case 2.} Suppose that $u(x)$ does not attain its infimum on $[c/\delta,\infty)$, that is, 
\begin{equation}\label{L3-2}
\inf_{x\in[c,c/\delta)} u(x)<\inf_{x\in[c/\delta,\bar x]}  u(x),
\end{equation}
as shown in Figs.~\ref{f:1}(b) and \ref{f:1}(c). Then the constraint in \eqref{E:Gamma1} is binding. So, for every $\eps>0$ there exists $(w_\eps,z_\eps,\sigma_\eps)$ with $w_\eps\le c/\delta\le z_\eps$ such that $F_{(w_\eps,z_\eps,\sigma_\eps)}$ satisfies \eqref{L3-1}, and satisfies  \eqref{L3-10} with equality.

As the solution lies on the convex closure of $u(x)$, the straight line through points $(w_\eps, u(w_\eps)-\eps)$ and $(z_\eps,u(z_\eps)-\eps)$ is weakly below the graph of $u$. Moreover, by  \eqref{L3-2}, the slope of this straight line is nonnegative, so $u(w_\eps)\le u(z_\eps)$. Thus, we obtain
\begin{equation}\label{E:wzineq}
\begin{split}
&u(w_\eps)-\eps\le u(x) \ \ \text{for all $x\ge w_\eps$,}\\
&u(z_\eps)-\eps\le u(x) \ \ \text{for all $x\ge z_\eps$}.
\end{split}
\end{equation}
As in Case 1, it follows that 
\[
u(z_\eps)-\eps\ge \frac{p(z_\eps) z_\eps-\eps}{1-\delta(1-p(z_\eps))}\ge U_p(F_{(w_\eps,z_\eps,\sigma_\eps)},z_\eps)-\frac{\eps}{1-\delta}.
\]
Also,
\begin{align*}
u(w_\eps)&\ge p(w_\eps)w_\eps+(1-p(w_\eps))\delta ((1-\sigma_\eps)(u(w_\eps)-\eps)+\sigma_\eps (u(z_\eps)-\eps))\\
&\ge p(w_\eps)w_\eps+(1-p(w_\eps))\delta ((1-\sigma_\eps)\left(u(w_\eps)-\eps)+\sigma_\eps  \frac{p(z_\eps) z_\eps-\eps}{1-\delta(1-p(z_\eps))}\right).
\end{align*}
Solving for $u(w_\eps)-\eps$, we obtain
\[
u(w_\eps)-\eps\ge \frac{p(w_\eps)w_\eps-\eps+(1-p(w_\eps))\delta\sigma_\eps  \frac{p(z_\eps) z_\eps-\eps}{1-\delta(1-p(z_\eps))}}{1-\delta(1-\sigma_{\eps})(1-p(w_\eps)}\ge U_p(F_{(w_\eps,z_\eps,\sigma_\eps)},\max\{y,w_\eps\})-\frac{\eps}{(1-\delta)^2}.
\]
We thus obtain
\begin{align*}
&U_p(G,y)\ge u(y)=p(y) y+(1-p(y))\delta \int_{0}^\infty u(x)\df G(x)\\
&\ge p(y) y+(1-p(y))\delta\left((1-\sigma_\eps)(u(w_\eps)-\eps)+\sigma_\eps  (u(z_\eps)-\eps)\right)\\ 
&\ge p(y) y+(1-p(y))\delta\left((1-\sigma_\eps)U_p(F_{(w_\eps,z_\eps,\sigma_\eps)},\max\{y,w_\eps\})+\sigma_\eps  U_p(F_{(w_\eps,z_\eps,\sigma_\eps)},z_\eps)\right)-\frac{\delta\eps}{(1-\delta)^2}\\ &=U_p(F_{(w_\eps,z_\eps,\sigma_\eps)},y)-\frac{\delta\eps}{(1-\delta)^2}.
\end{align*}
By \eqref{OptV} and the fact that the constraint in \eqref{E:Gamma1} is binding, observe that
\[
V(G,y)=\max\{y,c\}=V(F_{(w_\eps,z_\eps,\sigma_\eps)},y).
\]
Thus, for every $\eps>0$ there exists an environment $F_{(w_\eps,z_\eps,\sigma_\eps)}$ such that
\begin{equation}\label{E:Conc2}
\frac{U_p(G,y)}{V(F,y)}\ge \frac{U_p(G_{(w_\eps,z_\eps,\sigma_\eps)},y)-\frac{\delta\eps}{(1-\delta)^2}}{V(F_{(w_\eps,z_\eps,\sigma_\eps)},y)}.
\end{equation}
In particular, $F_{(w_\eps,z_\eps,\sigma_\eps)}$ satisfies  \eqref{L3-1} and \eqref{L3-10} if one replaces $\eps$ by $\delta\eps/(1-\delta)^2$.

\medskip
Taking $\eps\to 0$ in \eqref{E:Conc1} and \eqref{E:Conc2}, we conclude that, for each best-so-far alternative $y\ge x_0$ and each environment $G\in\mathcal F_X$,
\[
\frac{U_p(G,y)}{V(G,y)}\ge \inf_{(w,z,\sigma)} \frac{U_p(F_{(w,z,\sigma)},y)}{V(F_{(w,z,\sigma)},y)} \ \ \text{s.t. $w\le c/\delta\le z$, $\sigma\in[0,1]$, and $w,z\in X$.}
\]
It follows that
\[
R_p(x_0)\ge \inf_{y\ge x_0}\inf_{(w,z,\sigma)} \frac{U_p(F_{(w,z,\sigma)},y)}{V(F_{(w,z,\sigma)},y)} \ \ \text{s.t. $w\le c/\delta\le z$, $\sigma\in[0,1]$, and $w,z\in X$.}
\]
Thus we have shown that we can restrict attention to environments $F_{(w,z,\sigma)}$.

\medskip\noindent{\it Part 2.} We now show that we can further restrict the set of environments to binary environments $F_{(z,\sigma)}=F_{(0,z,\sigma)}$, so $w=0$. In other words, for each $y$, no environment $F_{(w,z,\sigma)}$ with $0<w<z$ and $\sigma<1$ can generate a payoff ratio $U_p(F_{(w,z,\sigma)},y)/V(F_{(w,z,\sigma)},y)$ smaller than $r_p(y)$ given by \eqref{E:r-p}.

By contradiction, suppose that there exists $y$ and an environment $F_{(w,z,\sigma)}$ with $0<w<z$ and $\sigma<1$ that gives a strictly smaller ratio than $r_p(y)$, so 
\begin{equation}\label{E:PPT0}
\frac{U_p(F_{(w,z,\sigma)},y)}{V(F_{(w,z,\sigma)},y)}<r_p(y).
\end{equation}
Let $c_{F_{(w,z,\sigma)}}$ denote the reservation value of $F_{(w,z,\sigma)}$ as defined by \eqref{E:RV}. 

Because the rule $p$ only depends on the best-so-far alternative, we only need to consider $w>y$, as otherwise $F_{(w,z,\sigma)}$ yields the same payoff as $F_{(0,z,\sigma)}$.
We have
\begin{equation}\label{E:PPT1-0}
U_p(F_{(w,z,\sigma)},y)=p(y)y+(1-p(y))\delta(\sigma U_p(F_{(w,z,\sigma)},z)+(1-\sigma) U_p(F_{(w,z,\sigma)},w))
\end{equation}
and
\begin{equation}\label{E:PPT1-1}
U_p(F_{(w,z,\sigma)},z)=U_p(F_{(0,z,1)},z).
\end{equation}

\medskip\noindent {\it Case 1}. Assume $U_p(F_{(w,z,\sigma)},w)\ge U_p(F_{(w,z,\sigma)},z)$ (see Figure~\ref{f:1}(a)). By \eqref{E:PPT1-0} and \eqref{E:PPT1-1},
\[
U_p(F_{(w,z,\sigma)},y)\ge p(y)y+(1-p(y)) U_p(F_{(w,z,\sigma)},z) = U_p(F_{(0,z,1)},y).
\]
Together with $V(F_{(w,z,\sigma)},y)\le V(F_{(0,z,1)},y)$, it follows that $\frac{U_p(F_{(w,z,\sigma)},y)}{V(F_{(w,z,\sigma)},y)}\ge \frac{U_p(F_{(0,z,1)},y)}{V(F_{(0,z,1)},y)}\ge r_p(y)$, which is a contradiction.

\medskip\noindent{\it Case 2}. Assume $U_p(F_{(w,z,\sigma)},w)<U_p(F_{(w,z,\sigma)},z)$ and $w>c_{F_{(w,z,\sigma)}}$ (see Figure~\ref{f:1}(b)). Then the optimal rule stops in the next round with certainty and, by \eqref{OptV}, yields the payoff of 
\[
V(F_{(w,z,\sigma)},y)=\delta((1-\sigma)w+\sigma z)=(1-\sigma)V(F_{(0,w,1)},y)+\sigma V(F_{(0,z,1)},y).
\]
Also, by \eqref{E:PPT1-0} and  \eqref{E:PPT1-1}, and using $U_p(F_{(w,z,\sigma)},z)>U_p(F_{(w,z,\sigma)},w)$, we obtain
\begin{align*}
U_p(F_{(w,z,\sigma)},w)&=p(y)y+(1-p(y))\delta((1-\sigma)U_p(F_{(w,z,\sigma)},w)+\sigma U_p(F_{(w,z,\sigma)},z))\\
&\ge p(y)y+(1-p(y))\delta U_p(F_{(w,z,\sigma)},w)=U_p(F_{(0,w,1)},y),
\end{align*}
and thus
\begin{align*}
U_p(F_{(w,z,\sigma)},y)&= p(y)y+(1-p(y))\delta((1-\sigma)U_p(F_{(w,z,\sigma)},w)+\sigma U_p(F_{(w,z,\sigma)},z))\\
&\ge (1-\sigma) U_p(F_{(0,w,1)},y)+\sigma U_p(F_{(0,z,1)},y).
\end{align*}
So,
\begin{align*}
\frac{U_p(F_{(w,z,\sigma)},y)}{V(F_{(w,z,\sigma)},y)}&\ge\frac{(1-\sigma)U_p(F_{(0,w,1)},y)+\sigma U_p(F_{(0,z,1)},y)}{(1-\sigma)V(F_{(0,z,1)},y)+\sigma V(F_{(0,w,1)},y)}\\
&\ge \min\left\{\frac{U_p(F_{(0,w,1)},y)}{V_p(F_{(0,w,1)},y)},\frac{U_p(F_{(0,z,1)},y)}{V_p(F_{(0,z,1)},y)}\right\}\ge r_p(y),
\end{align*}
which is a contradiction.

\medskip\noindent{\it Case 3}. Let $U_p(F_{(w,z,\sigma)},z)>U_p(F_{(w,z,\sigma)},w)$ and $w\le c_{F_{(w,z,\sigma)}}$ (see Figure~\ref{f:1}(c)). Then the optimal rule waits for the realization of $z$ and, by \eqref{OptV}, satisfies
\begin{equation}\label{E:PPT5}
V(F_{(w,z,\sigma)},y)=V(F_{(0,z,\sigma)},y).
\end{equation}
Rearranging \eqref{E:U-Z} we obtain
\begin{equation}\label{E:PPT1}
p(y)y+(1-p(y))\delta \sigma U_p(F_{(0,z,\sigma)},z)=(1-\delta(1-p(y))(1-\sigma))U_p(F_{(0,z,\sigma)},y).
\end{equation}
Thus,
\begin{align}\label{E:PPT2}
U_p(F_{(w,z,\sigma)},w)&=\frac{U_p(F_{(w,z,\sigma)},w)}{V(F_{(w,z,\sigma)},w)}V(F_{(w,z,\sigma)},w)=\frac{U_p(F_{(0,z,\sigma)},w)}{V(F_{(0,z,\sigma)},w)}V(F_{(w,z,\sigma)},w)\tag*{}\\
&\ge r_p(w) V(F_{(w,z,\sigma)},w)\ge r_p(y)V(F_{(w,z,\sigma)},w)\ge r_p(y)V(F_{(w,z,\sigma)},y)\notag\\
&> \frac{U_p(F_{(w,z,\sigma)},y)}{V(F_{(w,z,\sigma)},y)} V(F_{(w,z,\sigma)},y)=U_p(F_{(w,z,\sigma)},y),\label{E:PPT2}
\end{align}
where the first inequality is by the definition of $r_p(w)$, the second inequality is by the assumption that $r_p(y)$ is nondecreasing, the third inequality follows from \eqref{OptV}, and the fourth inequality is by \eqref{E:PPT0}. Then, using \eqref{E:PPT1-0}, \eqref{E:PPT1-1}, \eqref{E:PPT1}, and \eqref{E:PPT2}, we obtain
\[
U_p(F_{(w,z,\sigma)},y)>(1-\delta(1-\sigma)(1-p(y)))U_p(F_{(0,z,\sigma)},y)+\delta(1-\sigma)(1-p(y))U_p(F_{(w,z,\sigma)},y).
\]
Since $1-\delta(1-\sigma)(1-p(y))>0$, it follows that
\[
U_p(F_{(w,z,\sigma)},y)>U_p(F_{(0,z,\sigma)},y).
\]
Since $V(F_{(w,z,\sigma)},y)= V(F_{(0,z,\sigma)},y)$ by \eqref{E:PPT5}, we obtain
\[
\frac{U_p(F_{(w,z,\sigma)},y)}{V(F_{(w,z,\sigma)},y)}>\frac{U_p(F_{(0,z,\sigma)},y)}{V(F_{(0,z,\sigma)},y)}\ge r_p(y),
\]
which is a contradiction. 
This completes the proof.

\subsection{Proof of Theorem \ref{T:Quarter}}\label{s:tquarter}

{\it Part (a).} We need to show that the decision rule $\bar p$ given by the constant stopping probability
\begin{equation*}\label{E:pstar}
\bar p(y)=\bar \pi=\frac{1-\delta}{2-\delta} \ \ \text{for all $h_t$}
\end{equation*}
always yields a performance ratio of at least $1/4$. 

Note that $\bar p$ is stationary and monotone. It also has the monotone ratio property, as can be easily verified by substitution of $p(y)=p(z)=\bar\pi$ into \eqref{E:PR-1} and \eqref{E:PR-0}. By Proposition \ref{P:Binary}, we can restrict attention to binary environments in ${\mathcal B}_X$.

First, suppose that $c_{F_{(z,\sigma)}}\le y$. Using $p(y)=\bar\pi=(1-\delta)/(2-\delta)$, we have by \eqref{E:PR-0}
\begin{equation}\label{e-rat0}
\frac{U_{\bar p}(F_{(z,\sigma)},y)}{V(F_{(z,\sigma)},y)}\ge\frac{\bar\pi}{1-\delta(1-\bar\pi)}=\frac 1 2>\frac 1 4.
\end{equation}
Next, suppose that $c_{F_{(z,\sigma)}}>y$. By \eqref{E:PR-1}, using ${ p}(y)={ p}(z)=\bar\pi$, we obtain
\begin{align}
\frac{U_{\bar p}( F_{(z,\sigma)},y)}{V( F_{(z,\sigma)},y)}&=\frac{\bar\pi y+(1-\bar\pi)\delta\sigma \frac{\bar\pi z}{1-\delta(1-\bar\pi)}}{(1-\delta(1-\sigma)(1-\bar\pi))\frac{\delta\sigma z}{1-\delta(1-\sigma)}}> \frac{(1-\bar\pi)\frac{\bar\pi}{1-\delta(1-\bar\pi)}}{(1-\delta(1-\sigma)(1-\bar\pi))\frac{1}{1-\delta(1-\sigma)}} \notag\\
&\ge \frac{\bar\pi}{1-\delta (1-\bar\pi)}\left(1-\frac{\bar\pi}{1-\delta (1-\bar\pi)}\right)=\frac{1}{4},\label{e-rat1}
\end{align}
where the first inequality is by $y>0$, the second equality is by the minimum w.r.t.~$\sigma\in[0,1]$ attained at $\sigma=0$, and the last equality is by $\bar\pi=(1-\delta)/(2-\delta)$. 

Note that, the expression
\[
\frac{\bar\pi}{1-\delta (1-\bar\pi)}=\bar\pi+\delta(1-\bar\pi)\bar\pi+\delta^2(1-\bar\pi)^2\bar\pi+... 
\]
is the reciprocal of the expected delay of obtaining $z$ after its realization, and the expression $1-\frac{\bar\pi}{1-\delta (1-\bar\pi)}$ is the probability of not stopping before $z$ realizes for the first time. Setting $\frac{\bar\pi}{1-\delta (1-\bar\pi)}$ equal to $1/2$ maximizes the product, leading to a payoff ratio of $1/4$.

{\it Part (b).} Let $\sup X=\infty$. As shown above, the rule $\bar p$ yields $R_{\bar p}(x_0,\mathcal F_X)=1/4$. We now show that no other rule can achieve more than $1/4$, thus proving that $\bar p$ is dynamically robust.

By Proposition \ref{P:Stat}, we can restrict attention to decision rules that are stationary, monotone, and have the monotone ratio property. Consider any such rule $p(y)$. 

As $\sup X=\infty$, there exists an increasing sequence $(y_n)_{n\in\N}$ of elements in $X$ such that $y_1\ge x_0$ and $\lim_{n\to\infty} y_n=\infty$. Let $(z_k)_{k\in\N}$ be an increasing subsequence of $(y_n)_{n\in\N}$ and let $(\sigma_k)_{k\in \N}$ be a decreasing sequence of probabilities that satisfy the following. For all $k\in\N$,
\begin{equation}\label{E:wzs}
z_k\ge y_k,   \quad \lim_{k\to\infty} \sigma_k=0, \quad \text{and}\quad \lim_{k\to\infty} \frac {y_k}{\sigma_k z_k}=0,
\end{equation}
and, in addition,
\begin{equation}\label{E:vk}
c_k:=\frac{\delta \sigma_k z_k}{1-\delta(1-\sigma_k)}>y_k.
\end{equation}
Such sequences always exist, as $z_k$ can be chosen to increase fast enough relative to $y_k$. E.g., if $X=\R$, choose $\sigma_k=1/k$, $w_k=k$, and $z_k=k^3$. 

For each $k\in\N$, consider the binary environment $F_{(z_k,\sigma_k)}$. Let $y_k$ denote the best-so-far alternative. By \eqref{E:RVS}, $c_k$ is the reservation value for the environment $F_{(z_k,\sigma_k)}$. By \eqref{E:vk}, $c_k>y_k$, so the optimal rule waits for $z_k$ to realize. Therefore, by \eqref{E:PR-1},
\[
R_p(x_0,\mathcal F_X)\le \frac{U_p(F_{(z_k,\sigma_k)},y_k)}{V(F_{(z_k,\sigma_k)},y_k)}=\frac{\left(p(y_k)\frac{y_k}{\delta \sigma_k z_k}+(1-p(y_k)) \frac{p(z_k)}{1-\delta(1-p(z_k))}\right)(1-\delta(1-\sigma_k))}{1-\delta(1-\sigma_k)(1-p(y_k))}.
\]
As $p(y)$ is nondecreasing, and $y_k$ and $z_k$ diverge, both $p(y_k)$ and $p(z_k)$ converge to the same probability denoted by $\bar q$:
\[
\bar q=\lim_{k\to\infty}p(y_k)=\lim_{k\to\infty}p(z_k).
\]
Using the above and \eqref{E:wzs}, we obtain 
\begin{align*}
R_p(x_0,\mathcal F_X)\le \lim_{k\to\infty}\frac{U_p(F_{(z_k,\sigma_k)},w_k)}{V(F_{(z_k,\sigma_k)},w_k)}=\frac{(1-\bar q)\frac{\bar q}{1-\delta(1-\bar q)}(1-\delta)}{1-\delta(1-\bar q)}=\frac{(1-\delta)(1-\bar q)\bar q}{(1-\delta(1-\bar q))^2}\le \frac 1 4,
\end{align*}
where the last inequality is easily verified for $\delta\in(0,1)$ and $\bar q\in[0,1]$.

\subsection{Proof of Theorem \ref{P:Rules}}\label{s:prules}
Let $\sup X<\infty$. By rescaling the values, without loss of generality assume that
\[
\sup X=1.
\]
The proof consists of two steps. In Step 1, we assume that the set of feasible environments is  $\mathcal F_{\{0,1\}}$ and find a dynamically robust rule for each $x_0>0$. This rule will be different from the rule $q^*$ that we found in Section \ref{s:binary}. 
In Step 2, we expand the set of feasible environments to $\mathcal F_{X}$, where $X\subset[0,1]$ and $\{0,1\}\subset X$. We show that the previously derived rule attains the same performance ratio if and only if the outside option $x_0$ exceeds some constant $L>0$. We numerically find an upper bound for this constant, which is $1/89$.\footnote{Specifically, we fix $\delta$ and numerically (using Maple software) find the smallest value of $x_0\in(0,1]$ such that the payoff ratio over all environments in $\mathcal F_{[0,1]}$ is minimized by an environment that randomizes between 0 and 1. Let us call this value $L_\delta$. Thus, as long as $x_0\ge L_\delta$, the restriction to $\mathcal F_{\{0,1\}}$ is w.l.o.g. It turns out that the numerically calculated value $L_\delta$ is constant in $\delta$ and is approximately equal to (bounded from above by) $1/89$.}  

\medskip
{\it Step 1.} Fix $x_0\in(0,1]$, and denote
\[
r^*(x_0)=R^*(x_0,\mathcal B_{\{0,1\}}),
\]
where $R^*(x_0,\mathcal B_{\{0,1\}})$ is the dynamically robust ratio for binary environments $\mathcal B_{\{0,1\}}$ given by \eqref{E:R-Bin}. Note that $r^*(x_0)>1/2$ for all $x_0>0$. For binary environments $\mathcal B_{\{0,1\}}$ and rule $q$, recall that the payoff ratio $r_q(y)$ is given by
\[
 r_q(y)=\inf_{F\in\mathcal B_{\{0,1\}}}\frac{U_q(y,F)}{V(y,F)}, \ \ y\in[x_0,\delta r^*(x_0)).
\]
For each best-so-far alternative $y\in[x_0,1]$, we find the greatest probability of stopping, $p_{x_0}(y)$, under the constraint that the payoff ratio is equal to $r^*(x_0)$:
\begin{equation}\label{E:IneqR1}
p_{x_0}(y)=\max\{ q\in[0,1]: r_ q(y)\ge r^*(x_0)\}.
\end{equation}
Following steps \eqref{e-r0} and \eqref{e-r1} in the proof of Theorem \ref{T:Half}, for each $y\in[x_0,1]$ we have
\begin{align}
r_q(y)&=\min\left\{\frac{q}{1-\delta(1-q)},\min_{\sigma\in[0,1]}\frac{q y+(1-q)\delta \sigma}{(1-\delta(1-\sigma)(1-q))\frac{\delta\sigma}{1-\delta(1-\sigma)}}\right\}.\label{e-R}
\end{align}

Clearly, $p_{x_0}(y)=1$ for each $y\in[\delta r^*(x_0),1]$. Let $y\in[x_0,\delta r^*(x_0))$. Consider the two expressions under minimum in \eqref{e-R}. The first expression is strictly increasing in $q$, so it cannot be binding. The derivative of the second expression w.r.t.~$q$ has a constant sign for all $q$:
\begin{equation}\label{E:IneqR10}
\frac{\df}{\df q}\left(\frac{q y+(1-q)\delta\sigma}{1-\delta(1-q)(1-\sigma)}\right)=\frac{1-\delta(1-\sigma)}{(1-\delta(1-q)(1-\sigma))^2}\left(y-\frac{\delta\sigma}{1-\delta(1-\sigma)}\right).
\end{equation}
If \eqref{E:IneqR10} is nonnegative, then the solution of \eqref{E:IneqR1} is $p_{x_0}(y)=1$. If \eqref{E:IneqR10} is negative, then, if a solution of \eqref{E:IneqR1} exists, it must satisfy the equation
\begin{equation}\label{E:IneqR3}
\min_{\sigma\in[0,1]} \frac{ q y+(1- q)\delta\sigma}{\left(1-\delta(1- q)(1-\sigma)\right)\frac{\delta \sigma}{1-\delta(1-\sigma)}}= r^*(x_0).
\end{equation}
It is straightforward to verify that the unique solution $(\tilde q,\tilde \sigma)$ of \eqref{E:IneqR3} is given by
\begin{align*}
\tilde q&=\begin{cases}
\frac{(1-\delta)(1-r^*)(y+\sqrt{y r^*})(r^*+\sqrt{y r^*})}{(1-\delta)(1-r^*) 2yr^*+(\delta (r^*)^2+((1-\delta+y)y+(1-\delta-(3-\delta) y)r^*)\sqrt{yr^*}}, & \text{if $y\in (0,\delta (r^*)^2)$},\\
\frac{\delta(1-r^*)}{\delta-y}, & \text{if $y\in[\delta (r^*)^2,\delta r^*)$,}
\end{cases}\\
\tilde \sigma&=\begin{cases}
\frac{(1-\delta)(y+\sqrt{yr})}{\delta(r-y)}, & \text{if $y\in (0,\delta (r^*)^2)$,}\\
1, & \text{if $y\in[\delta (r^*)^2,\delta r^*)$,}
\end{cases}
\end{align*}
where we write $r^*$ for $r^*(x_0)$ for notational convenience.
It is also straightforward to verify that $\frac{\tilde q}{1-\delta(1-\tilde q)}\ge r^*(x_0)$, so $\tilde q$ is a solution of \eqref{E:IneqR1}.

So, for each $y\in[x_0,\delta r^*(x_0))$, we have $p_{x_0}(y)=\tilde q$ and, by construction, $r_{p_{x_0}}(y)=r^*(x_0)$. We thus obtain $R_{p_{x_0}}(x_0,\mathcal B_{\{0,1\}})=r^*(x_0)=R^*(x_0,\mathcal B_{\{0,1\}})$.

\medskip
{\it Step 2.} Now consider all environments in $\mathcal F_{X}$, where $X\subset [0,1]$ and $\{0,1\}\subset X$. As $\mathcal B_{\{0,1\}}\subset \mathcal F_{X}$, we have 
\[
R_{p_{x_0}}(x_0,\mathcal F_{X})\le R^*(x_0,\mathcal F_{X})\le r^*(x_0)=R^*(x_0,\mathcal B_{\{0,1\}}).
\] 

We now identify the lower bound $L$ on $x_0$ such that $R^*(x_0,\mathcal F_{X})=r^*(x_0)$ for all $x_0\in [L,1]$, and thus the rule $p_{x_0}$ derived in Step 1 is dynamically robust on $\mathcal F_{[0,1]}$.
Define
\[
L=\inf\{x_0\in(0,1]: R_{p_{x_0}}(x_0,\mathcal F_{X})=r^*(x_0)\}.
\] 
Observe that $R_{p_{x_0}}(x_0,\mathcal F_{X})=r^*(x_0)=1$ for all $x_0\in[\delta,1]$. However, $\lim_{x_0\to 0} r^*(x_0)=\rho(0)=1/2$ and $\lim_{x_0\to 0} R_{p_{x_0}}(x_0,\mathcal F_{X})\le 1/4$.\footnote{\label{f:12}Rescaling the values by $\bar x=1/x_0$, we have $\lim_{x_0\to 0} R^*(x_0,\mathcal F_{X})\le \lim_{\bar x\to\infty } R^*(1,\mathcal F_{[0,\bar x]})\le 1/4$ by Theorem \ref{T:Quarter}(b).} Therefore, $L\in(0,\delta]$. We numerically find that the value of $L$ is at most $1/89$,  with the equality when $X=[0,1]$. This numeric bound does not depend on the discount factor $\delta$.

It remains to show statements (a) and (b) of Theorem \ref{P:Rules}. By Theorem \ref{T:Half}, the rule $p^*$ satisfies $R_{p^*}(x_0,\mathcal B_{\{0,1\}})=\rho(x_0)=R^*(x_0,\mathcal B_{\{0,1\}})$ for all $x\le \delta^2/(2-\delta)$.
As $\mathcal B_{\{0,1\}}\subset \mathcal F_{X}$, we have for all $x\le \delta^2/(2-\delta)$
\[
R_{p^*}(x_0,\mathcal F_{X})\le R^*(x_0,\mathcal F_{X})\le \rho(x_0).
\] 
We now find the lower bound $L'$ on $x_0$ such that $R_{q^*}(x_0,\mathcal F_{X})\ge \rho(x_0)$ for all $x_0\in [L',1]$, and thus the rule $q^*$ is dynamically robust on $\mathcal F_{X}$ for $x_0\in[L',\delta^2/(2-\delta)]$. Define
\[
L'=\inf\{x_0\in(0,1]: R_{p^*}(x_0,\mathcal F_{X})\ge \rho(x_0)\}.
\] 
For $X=[0,1]$, we verify that $L'=1/6$, by checking that, for $x_0<1/6$, 
\[
\inf_{z\in[0,1],\sigma\in[0,1]}\frac{U_{p^*}(x_0,F_{(z,\sigma)})}{V(x_0,F_{(z,\sigma)})}<\rho(x_0).
\]
In words, for $x_0<1/6$, the worst-case ratio is attained by a lottery over $0$ and $z$ with $z<1$, which is why rule $p^*$ no longer attains the dynamically robust ratio $\rho(x_0)$.

\section*{Appendix C. Variations and Extensions}
\renewcommand{\thesection}{C}
\setcounter{subsection}{0}

\subsection{A Dynamically Robust Rule for Bounded Environments.\label{s:Deriv}}
In the following we present a recursive procedure for constructing a dynamically robust rule for any value of $x_0$. For simplicity, we consider the case where the set of alternatives $X$ is an interval. The proof is easily adapted to a more general case. 

By rescaling the values, without loss of generality assume that $X=[0,1]$. We fix a target performance ratio $r$ and find a rule, together with a threshold $x_0(r)$, such that this rule attains a performance ratio at least $r$ when the outside option $x_0$ is at least $x_0(r)$. We also show that there is no rule that has a performance ratio better than $r$ for $x_0=x_0(r)$, and use this to argue that $r$ is the dynamically robust performance ratio when $x_0=x_0(r)$.


Let us introduce the following notation. Let $q$ be a stationary and monotone decision rule such that $r_q(y)$ is nondecreasing. By \eqref{E:r-p}, \eqref{E:PR-1}, \eqref{E:PR-0}, and the fact that $y\le c_{F_{(z,\sigma)}}$ implies $y\le \delta z$, the payoff ratio of a rule that stops with probability $s\in[0,1]$ when the best-so-far alternative is $y$, and stops with probability $q(z)$ for all $z>y/\delta$ is given by
\begin{equation}\label{E:PQ}
\tilde r_q(y,s)=\min\left\{\frac{s}{1-\delta(1-s)},\inf_{\substack{\sigma\in[0,1],\\ z\in (y/\delta,1]}}\frac{sy+(1-s)\delta \sigma \frac{q(z) z}{1-\delta(1-q(z))}}{(1-\delta(1-s)(1-\sigma))\frac{\delta\sigma z}{1-\delta(1-\sigma)}}\right\}.
\end{equation}
Note that $\tilde r_q(y,q(y))=r_q(y)$ by the definition of $r_q$.

Fix a target performance ratio $r\in(\frac 1 4,1]$. The following procedure will derive a decision rule $p$ and a lower bound $x_0(r)$ such that $p$ attains the ratio $r_p(y)\ge r$ for all $y\in[x_0(r),1]$. This decision rule $p$ will be compared to a different hypothetical rule $q$ that guarantees a strictly better ratio at $x_0(r)$, and hence at all higher best-so-far alternatives. So we suppose that 
\begin{equation}\label{E:q-hyp}
r_q(y)>r \ \ \text{for all $y\in[x_0(r),1]\cap X$}.
\end{equation}
We will then show that no such $q$ exists, thus proving dynamic robustness of $p$.

We now construct $p$ by induction. During this construction, we will verify some properties of the hypothetical rule $q$.

First, for each $y\in[\delta,1]$ define
\[
S_r(y)=\left\{s\in [0,1]: \frac{s}{1-\delta(1-s)}\ge r\right\}
\]
and
\begin{equation}\label{E:popt0}
p(y)=\max\{s\in [0,1]:  s\in S_r(y)\}=1.
\end{equation}
By \eqref{E:PQ} and \eqref{E:q-hyp}, the hypothetical rule $q$ satisfies
\[
r_q(y)=\frac{q(y)}{1-\delta(1-q(y))}>r \ \ \text{for each $y\in[\delta,1]$},
\]
so $q(y)\in S_r(y)$.

We proceed by induction. For each $k=1,2,...$, we derive $p(y)$ for $y\in [\delta^{k+1},\delta^{k})$, using our solution $p(z)$ for all $z\ge \delta^k$ from the earlier induction steps. We also verify that $q(y)\in S_r(y) $ for each $y\in [\delta^{k+1},\delta^{k})$ using the induction assumption 
\begin{equation}\label{E:topt0}
q(z)\in S_r(z) \ \ \text{for all $z\in[\delta^k,1]$}.
\end{equation}

For each $y\in [\delta^{k+1},\delta^{k})$, define
\[
S_r(y)=\left\{s\in [0,1]: \tilde r_p(y,s)\ge r\right\}
\]
and
\begin{equation}\label{E:popt}
p(y)=\max\{s\in [0,1]:  s\in S_r(y)\}  \ \ \text{if $S_r(y)\ne\varnothing$.}
\end{equation}
Notice that $S_r(y)$ depends on $p$ only through the values of $p(z)$ defined in the previous iterations of the procedure.

Now we check the properties of  the hypothetical rule $q$. By \eqref{E:popt0} and \eqref{E:popt} and the induction assumption \eqref{E:topt0}, we have $q(z)\le p(z)$ for all $z\ge y/\delta$. By \eqref{E:PQ}, $\tilde r_q(y,s)$ is increasing in $q(z)$  for all $z\ge y/\delta$. Consequently,
\begin{equation}\label{E:topt-1}
\tilde r_q(y,s)\le \tilde r_p(y,s) \ \ \text{for all $s\in [0,1]$.}
\end{equation}
As $\tilde r_q(y,q(y))=r_q(y)>r$ by \eqref{E:q-hyp}, we obtain
\begin{equation}\label{E:topt}
q(y)\in S_r(y).
\end{equation}
Now let us return to the construction of $p$. If $S_r(y)$ is nonempty for all $y\in [\delta^{k+1},\delta^{k})$, then we proceed to the next step of the induction, $k+1$. Otherwise, we terminate the procedure. Upon termination, we define $S_r(y)=\varnothing$ for all $y\in (0,\delta^{k+1})$ and
\[
x_0(r)=\min\{y: S_r(y)\ne\varnothing\}.
\]
We thus obtain $p(y)$ that satisfies $r_{p}(y)\ge r$ for all $y\in[x_0(r),1]$.

Note that the procedure terminates in a finite number of steps for each $r>1/4$. The proof of Theorem \ref{T:Quarter}(b) actually shows that for each $\eps>0$ there exists $\bar x >0$ such that $R_p(1,\mathcal F_{ X})\ge 1/4+\eps$ whenever $\sup  X\le \bar x$. By rescaling the values by $1/(\sup X)$, we obtain that $R_p(x_0,\mathcal F_{X})\ge 1/4+\eps$ whenever $x_0\ge (\sup X)/\bar x$.

Furthermore, since $\tilde r_p(y,s)$ is continuous in $y$ and $s$, $S_r$ defined by the above procedure is continuous in $r$.\footnote{Specifically, the graph $\{(y,S_r(y))\}_{y>0}$ is continuous in $r$ in the topology of uniform convergence.} Therefore, $x_0(r)$ is continuous. 

We now show that every stopping probability in $S_r(x_0(r))$ gives the same payoff ratio, $r$, that is,
\begin{equation}\label{E:topt-2}
\tilde r_p(x_0(r),s)=r \ \ \text{for all $s\in  S_r(x_0(r))$}.
\end{equation}
If there were $s\in  S_r(x_0(r))$ such that $\tilde r_p(x_0(r),s)>r$, then, by continuity of $\tilde r_p(y,s)$ in $y$, there would exist $\eps>0$ such that $\tilde r_p(x_0(r)-\eps,s)\ge r$, which is a contradiction to the definition of $x_0(r)$. 

By continuity of  $x_0(r)$ and \eqref{E:topt-2}, we obtain that $x_0(r)$ is a one-to-one mapping. That is, for each $x_0>0$ there exists $r$, and a decision rule $p$ defined by \eqref{E:popt0} and \eqref{E:popt} for this value of $r$, such that $r_{p}(y)\ge r$ for all $y\in[x_0,1]$ with equality for $y=x_0$, and thus
\[
R_p(x_0,\mathcal F_X)=\inf_{y\ge x_0} r_p(y) =r.
\]

We now show that $p$ defined by \eqref{E:popt0} and \eqref{E:popt} for a given $r$ is dynamically robust. Recall the hypothetical rule $q$ that satisfies \eqref{E:q-hyp}. By \eqref{E:topt}, $q(y)\in S_r(y)$ for all $y\in[x_0(r),1]$. Inserting $s=q(x_0(r))$ into \eqref{E:topt-2} we obtain $\tilde r_p(x_0(r),q(x_0(r))=r$. By \eqref{E:topt-1}, 
\[
r_q(x_0(r))=\tilde r_q(x_0(r),q(x_0(r))\le \tilde r_p(x_0(r),q(x_0(r))=r.
\] 
This is a contradiction to \eqref{E:q-hyp}, thus proving dynamic robustness of $p$.

\subsection{Linear Decision Rules.}\label{s:linear}
Analogously to Appendix C.1, consider $X=[0,1]$. In this section we investigate how much we lose in terms of the performance ratio if we consider simple rules where the stopping probability is linear in the best-so-far alternative (wherever this probability is below 1). We find that these rules approximate the dynamically robust performance ratio identified in Theorem \ref{P:Rules} well, with the performance loss around 5\%, provided the discount factor is not too close to one.

Consider a truncated linear rule $p_{\alpha}$ given by
\[
p_{\alpha}\left(y\right) =\min \left \{ \tfrac{1-\delta }{2-\delta }+\alpha y,1\right \},
\]
where $\alpha>0$ is a parameter to be determined. Note that the
intercept $\frac{1-\delta }{2-\delta }$ is taken from decision rule $\bar p$ in Theorem \ref{T:Quarter}. 

The intercept ensures good performance when the best so far alternative $y$ is very small. The slope $\alpha$ is used to ensure good performance for higher values of $y$.

By Proposition \ref{P:Binary} it is sufficient to investigate performance when facing binary envirornments.
For each value of $\alpha$ and $x_0$, we derive the performance ratio $R_{p_\alpha}(x_0,\mathcal F_X)$ for the linear rule and evaluate the performance loss given by
\[
\eps_\alpha=\sup_{x_0\in (1/89,1)} \Big(R^*(x_0,\mathcal X)-R_{p_\alpha}(x_0,\mathcal X)\Big),
\]
where $R^*$ is the dynamically robust performance ratio (we consider $x_0\ge 1/89$ to apply Theorem \ref{P:Rules}). We search for the value $\alpha^{\ast}$ that minimizes the performance loss,
\[
\eps^*=\eps_{\alpha^*}=\inf\nolimits_{\alpha\ge 0} \eps_\alpha.
\]
That is, we look for linear rules that are closest to being dynamically robust. Closeness refers here to the smallest maximal loss in performance ratio, $\eps^*$, as compared to the dynamically robust rule.

Table \ref{T:2} presents, for various values of discount factor $\delta$, how much one loses in terms of the performance ratio when limiting attention to linear rules. It also presents the corresponding slopes of the linear rules.

\begin{table}[!htb]
\[
{
\begin{array}{l|llllllllllll}
\delta  & 0.1 & 0.2 & 0.3 & 0.4 & 0.5 & 0.6 & 0.7 & 0.8 & %
0.9 & 0.95 & 0.99 & 0.999 \\ 
\hline
\alpha^{\ast } & 4.68 & 2.48 & 1.75 & 1.38 & 1.19 & 1.05 & 0.93 & 
0.81 & 0.6 & 0.39 & 0.1 & 0.01 \\ 
\eps^{\ast } & 4.9\% & 4.7\% & 4.6\% & 4.5\% & 4.8\% & 5\% & %
4.8\% & 4.4\% & 5.5\% & 6.6\% & 8.1\% & 8.3\%
\end{array}}
\]
\caption{\small Numerically computed coefficients $\alpha^*$ of the linear rules that are closest to being dynamically robust, with the corresponding bounds  $\eps^*$.}\label{T:2}
\end{table}

One may not be satisfied by the performance of the linear rules when $\delta 
$ is large. For large $\delta $, there is a different simple rule that
performs almost as well as the dynamically robust rule. Let 
\[
\hat p_{\beta}(y)=\min \left \{ \sqrt{\frac{\beta \left( 1-\delta \right) y}{%
1-y}},1\right \}
\]
 for $1/89\le y<\delta $ and $\hat p_{\beta}\left(y\right) =1$ for $%
y\geq \delta ,$ where $\beta >0$ is a parameter. Again, we are searching for the parameter $\beta ^{\ast }$ that makes $\bar p_{\beta}$ closest to being dynamically robust. We list the values of $\beta ^{\ast }$ and the corresponding bounds on the performance loss in Table \ref{T:3}.

\begin{table}[!htb]
\[
\begin{array}{l|llll}
\delta  & 0.9 & 0.95 & 0.99 & 0.999 \\ 
\hline
\beta ^{\ast } & 1.35 & 0.8 & 0.22 & 0.024 \\ 
\eps^{\ast } & 2.8\% & 1.6\% & 2.5\% & 3\%
\end{array}
\]
\caption{\small Numerically computed coefficients $\beta^*$ of rules $\hat p_{\beta}$ that are closest to being dynamically robust, with the corresponding bounds  $\eps^*$.}\label{T:3}
\end{table}

\subsection{Additive-Multiplicative Search Costs.}\label{s:cost}

Let us now extend the model by introducing an additive cost of search. Suppose that the individual incurs a cost of $\kappa\ge 0$ in each round of search. That is, in each round $t$, the individual has a choice between consuming the best-so-far alternative $y_t$, or to proceed to the next round, where a fixed cost of $\kappa$ is incurred, and all future payoffs are discounted by $\delta$. Thus, if the individual stops the search in round $t\ge 1$, her payoff from the perspective of round 0 is
\[
-(\delta+...+\delta^{t-1}+\delta^t)\kappa+\delta^t y_t.
\]
We assume that the cost parameters satisfy $\kappa\ge 0$, $0<\delta\le 1$, and $\kappa+(1-\delta)>0$. The last assumption demands that the search is costly. We allow for either zero additive cost, $\kappa=0$, or zero multiplicative cost, $1-\delta=0$, but not both.

First, we point out that Propositions \ref{L:Mixed} and \ref{P:Hist}, as well as Propositions \ref{P:Stat} and \ref{P:Binary}, continue to hold in this setting. The proofs of these propositions are easily adjusted to take into account the additive cost of search.

As before, Propositions \ref{P:Stat} and \ref{P:Binary} allow us to restrict attention to monotone decision rules that depend on the best-so-far alternative only, and to narrow down the set of priors to the set of binary environments ${\mathcal B}_X$. 
Let $p$ be a monotone rule. Let $V(F,y)$ be the optimal payoff and $U_p(F,y)$ be the payoff of rule $p$ in environment $F$ under best-so-far alternative $y$. Then for $F\in\mathcal B_X$ we obtain
\[
V(F,y)=\max_{q\in[0,1]} \left(q y+(1-q)\delta\left(-\kappa+ \int_0^1 V(F,\max\{y,x\})\df F(x)\right)\right)
\]
and
\[
U_p(F,y)=p(y) y+(1-p(y))\delta \left(-\kappa+\int_0^1 U_p(F,\max\{y,x\})\df F(x)\right).
\]
The performance ratio of rule $p$ is defined for each outside option $x_0>0$ as
\[
R_p(x_0,\mathcal F_X)=\inf_{y\ge x_0}\inf_{ F\in{\mathcal B}_X} \frac{U_p( F,y)}{V( F,y)}.
\]

We now find the dynamically robust performance ratio when the outside option is at least twice the present value of all future discounted costs, so $x_0\ge \frac{2\delta\kappa}{1-\delta}$.

\begin{theorem}\label{T:C}
Let $x_0\ge \frac{2\delta\kappa}{1-\delta}$. The stationary decision rule $\bar p$ given by
\[
\bar p(y)=\frac{1-\delta}{2-\delta} \ \ \text{for all $y\ge x_0$}
\]

(a) attains the performance ratio at least $1/4$;

(b) is dynamically robust if $\sup X=\infty$.
%
\end{theorem}

\begin{proof}
Let $x_0\ge \frac{2\delta\kappa}{1-\delta}$ be an outside option. We show that the decision rule $\bar p(y)$ that stops with the constant probability
\[
q:=\bar p(y) \ \ \text{for all $y\ge x_0$}
\]
attains the performance ratio of $1/4$.

Fix a best-so-far alternative $y\ge x_0$. Consider an environment $F_{(z,\sigma)}\in{\mathcal B}_X$ such that $z\le y$, so $V(F_{(z,\sigma)},y)=y$, and
\[
U_{\bar p}(F_{(z,\sigma)},y)=q y+(1-q)\delta (U_p(F_{(z,\sigma)},y)-\kappa)=\frac{q y-(1-q)\delta\kappa}{1-\delta(1-q)}=\frac 1 2\left(y-\frac{\delta\kappa}{(1-\delta)}\right).
\]
By $y\ge x_0\ge \frac{2\delta\kappa}{1-\delta}$, the payoff ratio satisfies
\[
\frac{U_{\bar p}(F_{(z,\sigma)},y)}{V(F_{(z,\sigma)},y)}=\frac 1 2\left(1-\frac{\delta\kappa}{(1-\delta)y}\right)\ge \frac 1 4.
\]
Second, consider an environment $F_{(z,\sigma)}$ such that $z>y$, so
\[
V(F_{(z,\sigma)},y)=\max\left\{y,\frac{\delta(\sigma z-\kappa)}{1-\delta(1-\sigma)}\right\}.
\]
The individual's payoff is
\[
U_{\bar p}(F_{(z,\sigma)},y)=q y+(1-q)\delta \left((1-\sigma)U_{\bar p}(F_{(z,\sigma)},y)+\sigma U_{\bar p}(F_{(z,\sigma)},z)-\kappa\right).
\]
Substituting $q=\frac{1-\delta}{2-\delta}$ and $U_{\bar p}(F_{(z,\sigma)},z)=\frac 1 2\left(z-\frac{\delta\kappa}{(1-\delta)}\right)$, solving for $U_{\bar p}(F_{(z,\sigma)},y)$, and simplifying the expression yields
\[
U_{\bar p}(F_{(z,\sigma)},y)=\frac{q y+(1-q)\delta \left(\frac \sigma 2\left(z-\frac{\delta\kappa}{(1-\delta)}\right)-\kappa\right)}{1-\delta(1-q)(1-\sigma)}=\frac 1 2\left(y-\frac{\delta\kappa}{(1-\delta)}\right)+\frac{\delta\sigma(z-y)}{4(1-\delta)+2\delta\sigma}.
\]
If $V(F_{(z,\sigma)},y)=y$, then, by $y\ge x_0\ge \frac{2\delta\kappa}{1-\delta}$,
\begin{equation}\label{B-11}
\frac{U_{\bar p}(F_{(z,\sigma)},y)}{V(F_{(z,\sigma)},y)}=\frac 1 {2y}\left(y-\frac{\delta\kappa}{(1-\delta)}+\frac{\delta\sigma(z-y)}{2(1-\delta)+\delta\sigma}\right)\ge \frac 1 {2y}\left(y-\frac{\delta\kappa}{(1-\delta)}\right)\ge \frac 1 4. 
\end{equation}
If $V(F_{(z,\sigma)},y)=\frac{\delta(\sigma z-\kappa)}{1-\delta(1-\sigma)}$, then, by $y\ge x_0\ge \frac{2\delta\kappa}{1-\delta}$,
\begin{align}
\frac{U_{\bar p}(F_{(z,\sigma)},y)}{V(F_{(z,\sigma)},y)}&=\frac 1 2\left(y-\frac{\delta\kappa}{(1-\delta)}+\frac{\delta\sigma(z-y)}{2(1-\delta)+\delta\sigma}\right)\frac{1-\delta(1-\sigma)}{\delta(\sigma z-\kappa)}\notag\\
&\ge \frac 1 2\left(\frac{y}{2}+\frac{\delta\sigma(z-y)}{2(1-\delta)+\delta\sigma}\right)\frac{1-\delta(1-\sigma)}{\delta(\sigma z-\kappa)}.\label{B-12}
\end{align}
Assume that $\sigma\ge 2(1-\delta)/\delta$. Then the right-hand side in \eqref{B-12} is increasing in $z$, so reducing $z$ until $\frac{\delta(\sigma z-\kappa)}{1-\delta(1-\sigma)}=y$ yields $\inf_{z}\frac{U_{\bar p}(F_{(z,\sigma)},y)}{V(F_{(z,\sigma)},y)}\ge 1/4$ by \eqref{B-11}. Alternatively, assume that $\sigma< 2(1-\delta)/\delta$. Then the right-hand side in \eqref{B-12} is decreasing in $z$, so, taking $z\to\infty$, we obtain
\[
\frac{U_{\bar p}(F_{(z,\sigma)},y)}{V(F_{(z,\sigma)},y)}\ge \inf_{\sigma<\frac{2(1-\delta)}{\delta}}\lim_{z\to\infty}\frac{U_{\bar p}(F_{(z,\sigma)},y)}{V(F_{(z,\sigma)},y)}=\inf_{\sigma<\frac{2(1-\delta)}{\delta}}\frac{1-\delta(1-\sigma)}{4(1-\delta)+2\delta\sigma}= \frac 1 4,
\]
as can be easily verified for all $\sigma\in[0,1]$. We thus obtain
\[
\inf_{\sigma\in[0,1], z\ge 0}\frac{U_{\bar p}(F_{(z,\sigma)},y)}{V(F_{(z,\sigma)},y)}\ge \frac 1 4.
\]
Moreover, for a sequence $(z_k,\sigma_k)_{k\in\N}$ such that $z_k\to\infty$, $\sigma_k\to 0$, and $\frac{\delta(\sigma_k z_k-\kappa)}{1-\delta(1-\sigma_k)}> y$ for all $k$,
\[
\lim_{k\to\infty}\frac{U_{\bar p}(F_{(z_k,\sigma_k)},y)}{V(F_{(z_k,\sigma_k)},y)}= \frac 1 4.
\]
\end{proof}

{\setlength{\baselineskip}{0.19in} 
\setlength\bibsep{0.13\baselineskip}
\bibliographystyle{econometrica}
\bibliography{search}}

\end{document}